\documentclass[sigplan,10pt,screen]{acmart}
\usepackage{listings}
\usepackage{algorithm}
\usepackage[noend]{algorithmic}
\usepackage{mathtools}
\DeclarePairedDelimiter\ceil{\lceil}{\rceil}
\DeclarePairedDelimiter\floor{\lfloor}{\rfloor}

\definecolor{mygreen}{HTML}{467E7E}
\definecolor{mygray}{rgb}{0.5,0.5,0.5}
\definecolor{myblue}{HTML}{144B7D}
\definecolor{myorange}{HTML}{B25A00}
\lstdefinestyle{myC}{
  language=C,
  backgroundcolor=\color{white},
  basicstyle=\linespread{0.9}\ttfamily\small,
  breakatwhitespace=false,
  breaklines=true,
  commentstyle=[\bfseries\color{mygreen},    
  deletekeywords={...},
  escapeinside={<@}{@>},
  extendedchars=true,
  keepspaces=true,
  keywordstyle=\bfseries\color{myblue!70}, 
  otherkeywords={uint},     
  deletekeywords={get,angle, gamma, invariant},
  emph = { certiq_prove, match, end},
  emphstyle=\bfseries\color{myorange!70},
  showspaces=false, 
  showstringspaces=false,
  showtabs=false, 
  stringstyle=\color{mymauve},
  tabsize=1,
 morecomment=[f][\bfseries\color{mymauve!70}][0]{@},
 morecomment=[f][\bfseries\color{mygreen}][0]{//},
}

\copyrightyear{2020}
\acmYear{2020}
\setcopyright{acmlicensed}\acmConference[PLDI '20]{Proceedings of the 41st ACM SIGPLAN International Conference on Programming Language Design and Implementation}{June 15--20, 2020}{London, UK}
\acmBooktitle{Proceedings of the 41st ACM SIGPLAN International Conference on Programming Language Design and Implementation (PLDI '20), June 15--20, 2020, London, UK}
\acmPrice{15.00}
\acmDOI{10.1145/3385412.3385986}
\acmISBN{978-1-4503-7613-6/20/06}

\usepackage{booktabs}   
\usepackage{subcaption} 
                        
\newcommand{\parheader}[1]{\noindent {\bf #1}}
\usepackage{amssymb}
\usepackage{amsmath}
\usepackage{amsthm}
\usepackage{pifont}
\newcommand{\cmark}{\ding{51}}%
\newcommand{\xmark}{\ding{55}}%
\usepackage{siunitx}
\newcommand{\specialcellbold}[2][b]{%
  \bfseries
  \sisetup{text-rm=\bfseries}%
  \begin{tabular}[#1]{@{}c@{}}#2\end{tabular}%
}

\definecolor{mred}{rgb}{.80,.12,.30}
\definecolor{grey}{rgb}{0.5,0.5,0.5}
\definecolor{purple2}{rgb}{.75,0,.85}
\definecolor{pistachio}{rgb}{0.58, 0.77, 0.45}
\definecolor{steelblue}{rgb}{.10,.40,.85}

\newcommand{\gabe}[1] {\textcolor{purple2}{(Gabe: #1)}}

\newcommand{\revise}[1]  {\textcolor{black}{#1}}

\renewcommand{\C}[0]{\mathcal{S}}

\begin{document}

\title[Learning Nonlinear Loop Invariants with G-CLNs]{Learning Nonlinear Loop Invariants with Gated \\ Continuous Logic Networks (Extended Version)}         


\newcommand*\samethanks[1][\value{footnote}]{\footnotemark[#1]}
\author{Jianan Yao}
\authornote{Equal contribution}
\affiliation{Columbia University\country{USA}}
\email{jianan@cs.columbia.edu}          
\author{Gabriel Ryan}
\authornotemark[1]
\affiliation{Columbia University\country{USA}}
\email{gabe@cs.columbia.edu} 
\author{Justin Wong}
\authornotemark[1]
\affiliation{Columbia University\country{USA}}
\email{justin.wong@columbia.edu} 
\author{Suman Jana}
\affiliation{Columbia University\country{USA}}
\email{suman@cs.columbia.edu} 
\author{Ronghui Gu}
\affiliation{Columbia University, CertiK\country{USA}}
\email{rgu@cs.columbia.edu} 
\renewcommand{\shortauthors}{J. Yao, G. Ryan, J. Wong, S. Jana, R. Gu}

\newcommand\blfootnote[1]{%
  \begingroup
  \renewcommand\thefootnote{}\footnote{#1}%
  \addtocounter{footnote}{-1}%
  \endgroup
}

\begin{abstract}
  \revise{Verifying real-world programs often requires inferring loop invariants with nonlinear constraints.} This is especially true in programs that perform many numerical operations, such as control systems for avionics or industrial plants. Recently, data-driven methods for loop invariant inference have shown promise, especially on linear loop invariants. However, applying data-driven inference to nonlinear loop invariants is challenging due to the large numbers of and large magnitudes of high-order terms, the potential for overfitting on a small number of samples, and the large space of possible nonlinear inequality bounds.


In this paper, we introduce a new neural architecture for general SMT learning, the Gated Continuous Logic Network (G-CLN), and apply it to nonlinear loop invariant learning. G-CLNs extend the Continuous Logic Network (CLN) architecture with gating units and dropout, which allow the model to robustly learn general invariants over large numbers of terms. To address overfitting that arises from finite program sampling, we introduce fractional sampling---a sound relaxation of loop semantics to continuous functions that facilitates unbounded sampling on the real domain. We additionally design a new CLN activation function, the Piecewise Biased Quadratic Unit (PBQU), for naturally learning tight inequality bounds. 

We incorporate these methods into a nonlinear loop invariant inference system that can learn general nonlinear loop invariants. We evaluate our system on a benchmark of nonlinear loop invariants and show it solves 26 out of 27 problems, 3 more than prior work, with an average runtime of 53.3 seconds. We further demonstrate the generic learning ability of G-CLNs by solving all 124 problems in the linear Code2Inv benchmark. We also perform a quantitative stability evaluation and show G-CLNs have a convergence rate of $97.5\%$ on quadratic problems, a $39.2\%$ improvement over CLN models. 



 
\end{abstract}



 \maketitle

\section{Introduction}


    
        

Formal verification provides techniques for proving the correctness of programs, thereby eliminating entire classes of critical bugs. 
While many operations can be verified automatically, verifying programs with loops usually requires inferring a sufficiently strong loop invariant, which is undecidable in general~\cite{hoare1969axiomatic,blass2001inadequacy, furia2014loop}. Invariant inference systems are therefore based on heuristics that work well for loops that appear in practice.
 Data-driven loop invariant inference is one approach that has shown significant promise, especially for learning linear invariants \cite{zhu2018chc, si2018learning, cln2inv}. Data-driven inference operates by sampling program state across many executions of a program and trying to identify an Satisfiability Modulo Theories (SMT) formula that is satisfied by all the sampled data points. 

However, verifying real-world programs often requires loop invariants with nonlinear constraints. This is especially true in programs that perform many numerical operations, such as control systems for avionics or industrial plants~\cite{damm2005guaranteed, lin2014hybrid}. 
Data-driven nonlinear invariant inference is fundamentally difficult because the space of possible nonlinear invariants is large, but sufficient invariants for verification must be inferred from a finite number of samples.
In practice, this leads to three distinct challenges when performing nonlinear data-driven invariant inference: (i) {\it Large search space with high magnitude terms.} Learning nonlinear terms causes the space of possible invariants to grow quickly (i.e. polynomial expansion of terms grows exponentially in the degree of terms). Moreover, large terms such as $x^2$ or $x^y$ dominate the process and prevent meaningful invariants from being learned.
(ii) {\it Limited samples.} Bounds on the number of loop iterations in programs with integer variables limit the number of possible samples, leading to overfitting when learning nonlinear invariants.  (iii) {\it Distinguishing sufficient inequalities.} For any given finite set of samples, there are potentially infinite valid inequality bounds on the data. However, verification usually requires specific bounds that constrain the loop behavior as tightly as possible.

\begin{figure*}[t]
\vspace{-10pt}
\begin{subfigure}[b]{0.48\textwidth}
\quad
\begin{minipage}{.45\linewidth}
\begin{lstlisting}[style=myC,numbers=none,basicstyle=\linespread{0.9}\ttfamily\footnotesize]
// pre: (a >= 0)
n=0; x=0;
y=1; z=6;
// compute cube:
while(n != a){
   n += 1;
   x += y;
   y += z;
   z += 6;
}
return x;
// post: x == a^3
\end{lstlisting}
\end{minipage}
\hspace{-30pt}
\begin{minipage}{.6\linewidth}
\vspace{15pt}
\includegraphics[width=.95\linewidth]{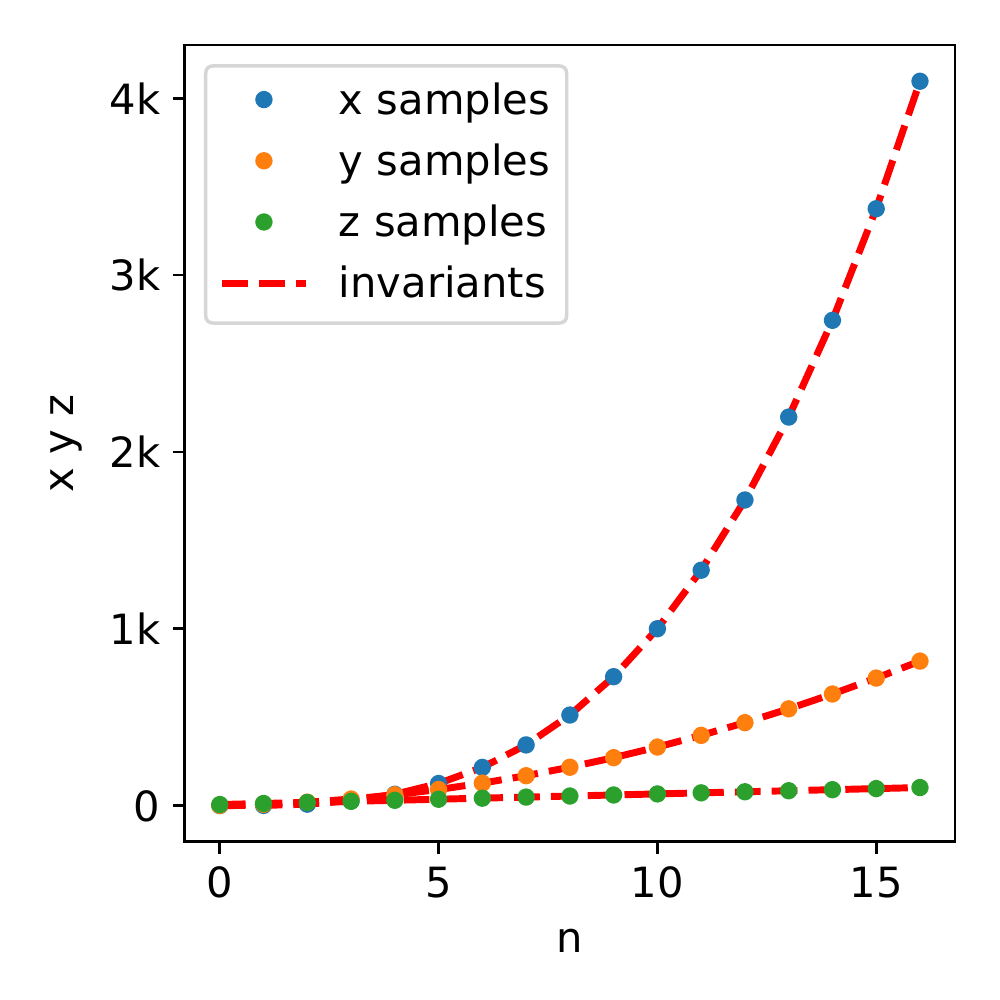}
\end{minipage}
\setlength{\abovecaptionskip}{-5pt}
\caption{\label{fig:cohencu_ex} Loop for computing cubes that requires the invariant $(x=n^3)\land (y=3n^2+3n+1)\land (z=6n+6)$ to infer its postcondition $(x = a^3)$. A data-driven model must simultaneously learn a cubic constraint that changes by 1000s and a linear constraint that increments by 6.}
\end{subfigure}
\quad \;\;
\begin{subfigure}[b]{0.48\textwidth}
\quad
\begin{minipage}{.45\linewidth}
\hspace{-5pt}
\begin{lstlisting}[style=myC,numbers=none,basicstyle=\linespread{0.9}\ttfamily\footnotesize]
// pre: (n >= 0)
a=0; s=1; t=1;
// compute sqrt:
while (s <= n) {
  a += 1;
  t += 2;
  s += t;
}
return a;
//post: a^2 <= n
//and n < (a+1)^2 
\end{lstlisting}
\end{minipage}
\hspace{-30pt}
\begin{minipage}{.6\linewidth}
\vspace{25pt}
\includegraphics[width=.95\linewidth]{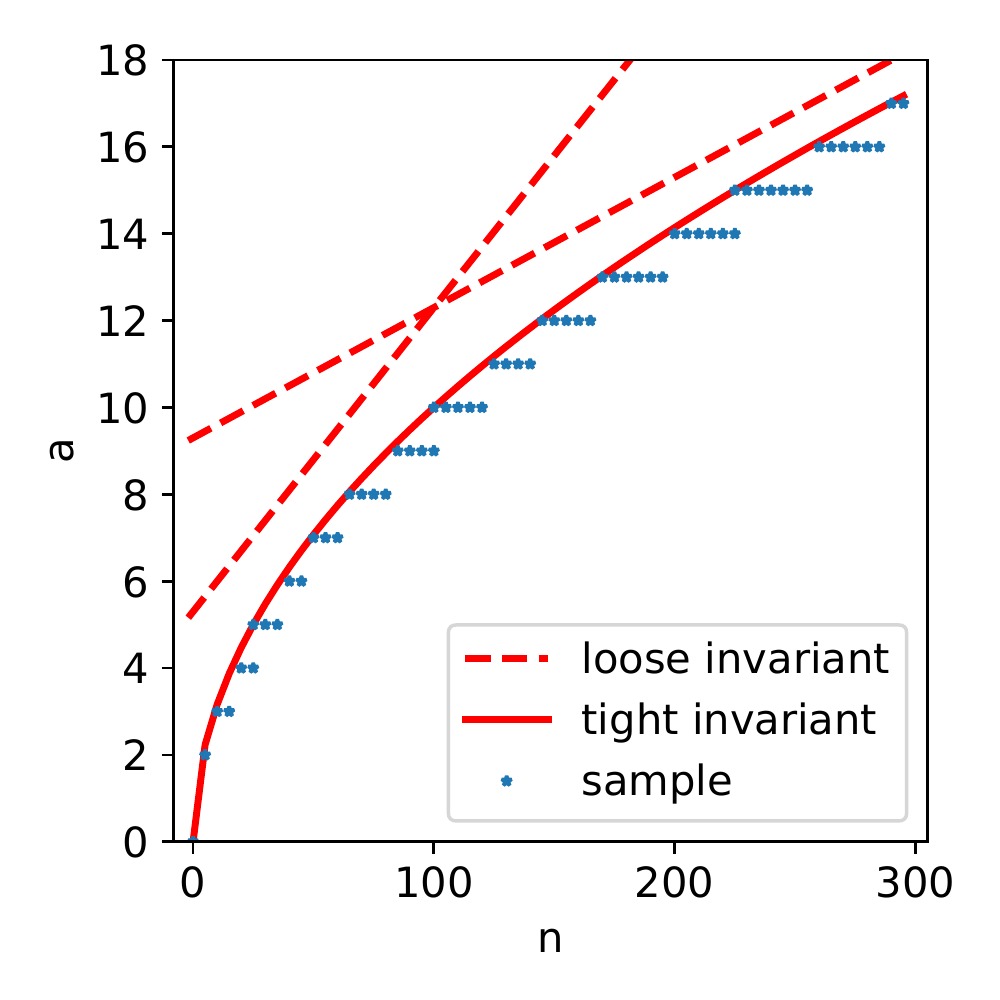}
\end{minipage}
\setlength{\abovecaptionskip}{-5pt}
\caption{\label{fig:ineq_ex} Loop for computing integer approximation to square root. The graph shows three valid inequality invariants, but only the tight quadratic inequality invariant $(n \geq a^2)$ is sufficient to verify that the final value of $a$ is between $\floor{ sqrt(n)}$ and $\ceil{ sqrt(n)}$. }
\end{subfigure}
\caption{\label{fig:example_problems}Example problems demonstrating the challenges of nonlinear loop invariant learning. }
\end{figure*}

Figure \ref{fig:cohencu_ex} and \ref{fig:ineq_ex} illustrate the challenges posed by loops with many higher-order terms as well as nonlinear inequality bounds. The loop in Figure~\ref{fig:cohencu_ex} computes a cubic power and requires the invariant $(x=n^3)\land (y=3n^2+3n+1)\land (z=6n+6)$ to verify its postcondition $(x=a^3)$. To infer this invariant, a \revise{typical} data-driven inference system must consider 35 possible terms, ranging from $n$ to $x^3$, only seven of which are contained in the invariant. Moreover, the 
higher-order terms in the program will dominate any error measure in fitting an invariant, so any data-driven model will tend to only learn the constraint $(x=n^3)$. Figure \ref{fig:ineq_ex} shows a loop for computing integer square root where the required invariant is $(n \geq a^2)$ to verify its postcondition. However, a data-driven model must identify this invariant from potentially infinite other valid but loosely fit inequality invariants.

Most existing methods for nonlinear loop invariant inference address these challenges by limiting either the structure of invariants they can learn or the complexity of invariants they can scale to. Polynomial equation  solving methods such as Numinv and Guess-And-Check are able to learn equality constraints but cannot learn nonlinear inequality invariants~\cite{nguyen2017numinv, sharma2013data}. In contrast, template enumeration methods such as PIE can potentially learn arbitrary invariants but struggle to scale to loops with nonlinear invariants because space of possible invariants grows too quickly \cite{padhi2016data}. 

In this paper, we introduce an approach 
that can learn general nonlinear loop invariants.
Our approach is based on Continuous Logic Networks (CLNs), a recently proposed neural architecture that can learn SMT formulas directly from program traces \cite{cln2inv}. 
CLNs use a parameterized relaxation that relaxes SMT formulas to differentiable functions.
This allows CLNs to learn SMT formulas with gradient descent, but a template that defines
the logical structure of the formula has to be manually provided.

We base our approach on three developments that address the challenges inherent in nonlinear loop invariant inference: First, we introduce a new neural architecture, the \emph{Gated Continuous Logic Network} (G-CLN), a more robust CLN architecture that is not dependent on formula templates. Second, we introduce \emph{Fractional Sampling}, a principled program relaxation for dense sampling. Third, we derive the \emph{Piecewise Biased Quadratic Unit} (PBQU), a new CLN activation function for inequality learning. We provide an overview of these methods below.


\parheader{Gated Continuous Logic Networks.} 
G-CLNs improve the CLN architecture by making it more robust and general. Unlike CLNs, G-CLNs are not dependent on formula templates for logical structure. We adapt three different methods from deep learning to make G-CLN training more stable and combat overfitting:
gating, dropout, and batch normalization~\cite{srivastava2014dropout, ioffe2015batch, gers1999learning, bahdanau2014neural}. To force the model to learn a varied combination of constraints, we apply {\it Term Dropout}, which operates similarly to dropout in feedforward neural networks by zeroing out a random subset of terms in each clause. Gating makes the CLN architecture robust by allowing it to ignore subclauses that cannot learn satisfying coefficients for their inputs, due to poor weight initialization or dropout.  
To stabilize training in the presence of high magnitude nonlinear terms, we apply normalization to the input and weights similar to batch normalization.

By combining dropout with gating, G-CLNs are able to learn complex constraints for loops with many higher-order terms. For the loop in Figure~\ref{fig:cohencu_ex}, the G-CLN will set the $n^2$ or $n^3$ terms to zero in several subclauses during dropout, forcing the model to learn a conjunction of all three equality constraints. Clauses that cannot learn a satisfying set of coefficients due to dropout, i.e. a clause with only $x$ and $n$ but no $n^3$ term, will be ignored by a model with gating.

\parheader{Fractional Sampling.}
When the samples from the program trace are insufficient to learn the correct invariant due to bounds on program behavior, we perform a principled relaxation of the program semantics to continuous functions. This allows us to  perform Fractional Sampling, which generates samples of the loop behavior at intermediate points between integers. To preserve soundness, we define the relaxation such that operations retain their discrete semantics relative to their inputs but operate on the real domain, and any invariant for the continuous relaxation of the program must be an invariant for the discrete program.
This allows us to take potentially unbounded samples even in cases where the program constraints prevent sufficient sampling to learn a correct invariant. 


\parheader{Piecewise Biased Quadratic Units.}
For inequality learning, we design a PBQU activation, which penalizes loose fits and converges to tight constraints on data. We prove this function will learn a tight bound on at least one point and demonstrate empirically that it learns precise bounds invariant bounds, such as the $(n \geq a^2)$ bound shown in Figure~\ref{fig:ineq_ex}.

We use G-CLNs with Fractional Sampling and PBQUs to develop a unified approach for general nonlinear loop invariant inference. 
 We evaluate our approach on a set of loop programs with nonlinear invariants, and show it can learn invariants for 26 out of 27 problems, 3 more than prior work, with an average runtime of 53.3 seconds. We also perform a quantitative stability evaluation and show G-CLNs have a convergence rate of $97.5\%$ on quadratic problems, a $39.2\%$ improvement over CLN models. We also test the G-CLN architecture on the linear Code2Inv benchmark \cite{si2018learning} and show it can solve all 124 problems.

In summary, this paper makes the following contributions:
\begin{itemize}
    \item We develop a new general and robust neural architecture, the Gated Continuous Logic Network (G-CLN), to learn general SMT formulas without relying on formula templates for logical structure.
    \item We introduce Fractional Sampling, a method that facilitates sampling on the real domain by applying a principled relaxation of program loop semantics to continuous functions while preserving soundness of the learned invariants. 
    \item We design PBQUs, a new activation function for learning tight bounds inequalities, and provide convergence guarantees for learning a valid bound.
    \item We integrate our methods in a general loop invariant inference system and show it solves 26 out 27 problems in a nonlinear loop invariant benchmark, 3 more than prior work. Our system can also infer loop invariants for all 124 problems in the linear Code2Inv benchmark. 
\end{itemize}

The rest of the paper is organized as follows:
In \S \ref{sec:background}, we provide background on the loop invariant inference problem, differentiable logic, and the CLN neural architecture for SMT learning. Subsequently, we introduce the high-level workflow of our method in \S \ref{sec:workflow}. 
Next, in \S \ref{sec:theory}, we formally define the gated CLN construction, relaxation for fractional sampling, and PBQU for inequality learning, and provide soundness guarantees for gating and convergence guarantees for bounds learning. We then provide a detailed description of our approach for nonlinear invariant learning with CLNs in \S \ref{sec:methodology}. Finally we show evaluation results in \S \ref{sec:evaluation} and discuss related work in \S \ref{sec:related} before concluding in \S \ref{sec:conclusion}.

\section{Background}
\label{sec:background}

In this section,
we provide a brief review of the loop invariant inference problem and then define the differentiable logic operators and the Continuous Logic Network architecture used in our approach.

\subsection{Loop Invariant Inference}
\label{section:background_loopinvariants}


Loop invariants encapsulate properties of the loop which are independent of the iterations and enable verification to be performed over loops. For an invariant to be sufficient for verification, it must simultaneously be weak enough to be derived from the precondition and strong enough to conclude the post-condition.
\revise{Formally, the loop invariant inference problem is, given a loop ``\texttt{while(LC) C},'' a precondition $P$, and a post-condition $Q$, we are asked to find an inductive invariant $I$ that satisfies the following three conditions}:
\begin{align*}
    P \implies I & & \{I \land LC\}\ C\ \{I\} & & I \land \neg LC \implies Q 
\end{align*}
where the inductive condition is defined using a Hoare triple. 

Loop invariants can be encoded in SMT, which facilitates efficient checking of the conditions with solvers such as Z3~\cite{de2008z3, handbookofSAT}. As such, our work focuses on inferring likely candidate invariant as validating a candidate can be done efficiently.

\parheader{Data-driven Methods.} Data-driven loop invariant inference methods use program traces recording the state of each variable in the program on every iteration of the loop to guide the invariant generation process. Since an invariant must hold for any valid execution, the collected traces can be used to rule out many potential invariants. Formally, given a set of program traces $X$, data-driven invariant inference finds SMT formulas $F$ such that:
$$ \forall x \in X,  F(x) = True$$


\subsection{Basic Fuzzy Logic}
\label{section:background_BL}
Our approach to SMT formula learning is based on a form of differentiable logic called Basic Fuzzy Logic (BL).
BL is a relaxation of first-order logic that operates on continuous truth values on the interval $[0, 1]$ instead of on boolean values. BL uses a class of functions called {\it t-norms} ($\otimes$), which preserves the semantics of boolean conjunctions on continuous truth values. T-norms are required to be consistent with boolean logic, monotonic on their domain, commutative, and associative \cite{hajek2013metamathematics}. Formally, a t-norm is defined $\otimes: [0, 1]^2 \rightarrow [0,1]$ such that:
\begin{itemize}
    \item $\otimes$ is consistent for any $t \in [0,1]$:
    \begin{align*}
        t \otimes 1 &= t&  t\otimes 0 &= 0
    \end{align*}
    \item $\otimes$ is commutative and associative for any $t \in [0, 1]$:
    \begin{align*}
        t_1 \otimes t_2 &= t_2 \otimes t_1&
        t_1 \otimes (t_2 \otimes t_3) &= (t_1 \otimes t_2) \otimes t_3
    \end{align*}
       
    \item $\otimes$ is monotonic (nondecreasing) for any  $t \in [0,1]$:
    \begin{align*}
        t_1 \leq t_2 \implies t_1 \otimes t_3 \leq t_2 \otimes t_3
    \end{align*}
\end{itemize}
BL additionally requires that t-norms be continuous. {\it T-conorms} ($\oplus$) are derived from t-norms via DeMorgan's law and operate as disjunctions on continuous truth values, while negations are defined $\neg t := 1 - t$.

\revise{In this paper, we keep t-norms abstract in our formulations to make the framework general. Prior work~\cite{cln2inv} found product t-norm $x \otimes y = x \cdot y$ perform better in Continuous Logic Networks. For this reason, we use product t-norm in our final implementation, although other t-norms (e.g., Godel) can also be used. }

\subsection{Continuous Logic Networks}
\label{section:background_CLNs}

We perform SMT formula learning with Continuous Logic Networks (CLNs),
 a neural architecture introduced in \cite{cln2inv} that are able to learn SMT formulas directly from data. These can be used to learn loop invariants from the observed behavior of the program. 

CLNs are based on a parametric relaxation of SMT formulas that maps the SMT formulation from boolean first-order logic to BL. The model defines the operator $\C$. Given an quantifier-free SMT formula $F: X \rightarrow \{True, False\}$,
$\C$ maps it to a continuous function $\C(F): X \rightarrow [0, 1]$.  In order for the continuous model to be both 
usable in gradient-guided optimization 
while also preserving the semantics of boolean logic, 
it must fulfill three conditions:
\begin{enumerate}
  \item It must preserve the meaning of the logic, such that the continuous truth values of a valid assignment are always greater than the value of an invalid assignment:
  \begin{align*}
      (F(x) = True \wedge F(x') &= False) \\\implies{}
      &  (\C(F)(x) > \C(F)(x'))
  \end{align*}
  \item It must be must be continuous and smooth (i.e. differentiable almost everywhere) to facilitate training.
  \item It must be strictly increasing as an unsatisfying assignment of terms approach satisfying the mapped formula, and strictly decreasing as a satisfying assignment of terms approach violating the formula. 
\end{enumerate}
$\C$ is constructed as follows to satisfy these requirements. The logical relations $\{\land, \lor, \neg\}$ are mapped to their continuous equivalents in BL:
\begin{align*}
&\textrm{Conjunction:}& \C(F_1 \land F_2)  &\triangleq  \C(F_1) \otimes \C(F_2)\\
&\textrm{Disjunction:}& \C(F_1 \lor F_2) &\triangleq \C(F_1) \oplus \C(F_2)\\
&\textrm{Negation:}& \C(\neg F)  &\triangleq 1 - \C(F)
\end{align*}
\noindent where any $F$ is an SMT formula. $\C$ defines SMT predicates $\{=, \neq, <, \leq, >, \geq\}$ with functions that map to continuous truth values. This mapping is defined for $\{>, \geq\}$ using sigmoids with a shift parameter $\epsilon$ and smoothing parameter $B$: 
\begin{align*}
    &\textrm{Greater Than:}& &\C(x_1 > x_2)  \triangleq \frac{1}{1 + e^{-B(x_1-x_2-\epsilon)}}\\
    &\textrm{Greater or Equal to:}& &\C(x_1 \geq x_2)  \triangleq \frac{1}{1 + e^{-B(x_1-x_2+\epsilon)}}
\end{align*}
\noindent where $x_1, x_2 \in R$.  Mappings for other predicates are derived from their logical relations to $\{>, \geq\}$:
\begin{align*}
    &\textrm{Less Than:}& &\C(x_1 < x_2) = \C(\neg(x_1 \geq x_2))\\
    &\textrm{Less or Equal to:}& &\C(x_1 \leq x_2) = \C(\neg (x_1 > x_2))\\
    &\textrm{Equality:}& &\C(x_1 = x_2) = \C((x_1 \geq x_2) \land (x_1 \leq x_2))\\
    &\textrm{Inequality:}& &\C(x_1 \neq x_2) = \C(\neg (x_1=x_2))
\end{align*}
Using these definitions the parametric relaxation $\C$ satisfies all three conditions for sufficiently large $B$ \revise{and sufficiently small $\epsilon$}. Based on this parametric relaxation $S(F)$, we build a Continuous Logic Network model $M$, which is a computational graph of $S(F)(x)$ with learnable parameters $W$. When training a CLN, loss terms are applied to penalize small $B$, ensuring that as the loss approaches 0 the CLN will learn a precise formula. Under these conditions, the following relationship holds between a trained CLN model $M$ with coefficients $W$ and its associated formula $F$ for a given set of data points, $X$:
$$\forall x \in X, M(x;W) = 1 \iff F(x; W) = True$$ 
\noindent Figure \ref{fig:ex_cln_plot} shows an example CLN for the formula on a single variable $x$: $$F(x) = (x = 1) \lor  (x \geq 5) \lor(x \geq 2 \land x \leq 3)$$

\begin{figure}
  \centering
    \includegraphics[width=.9\linewidth]{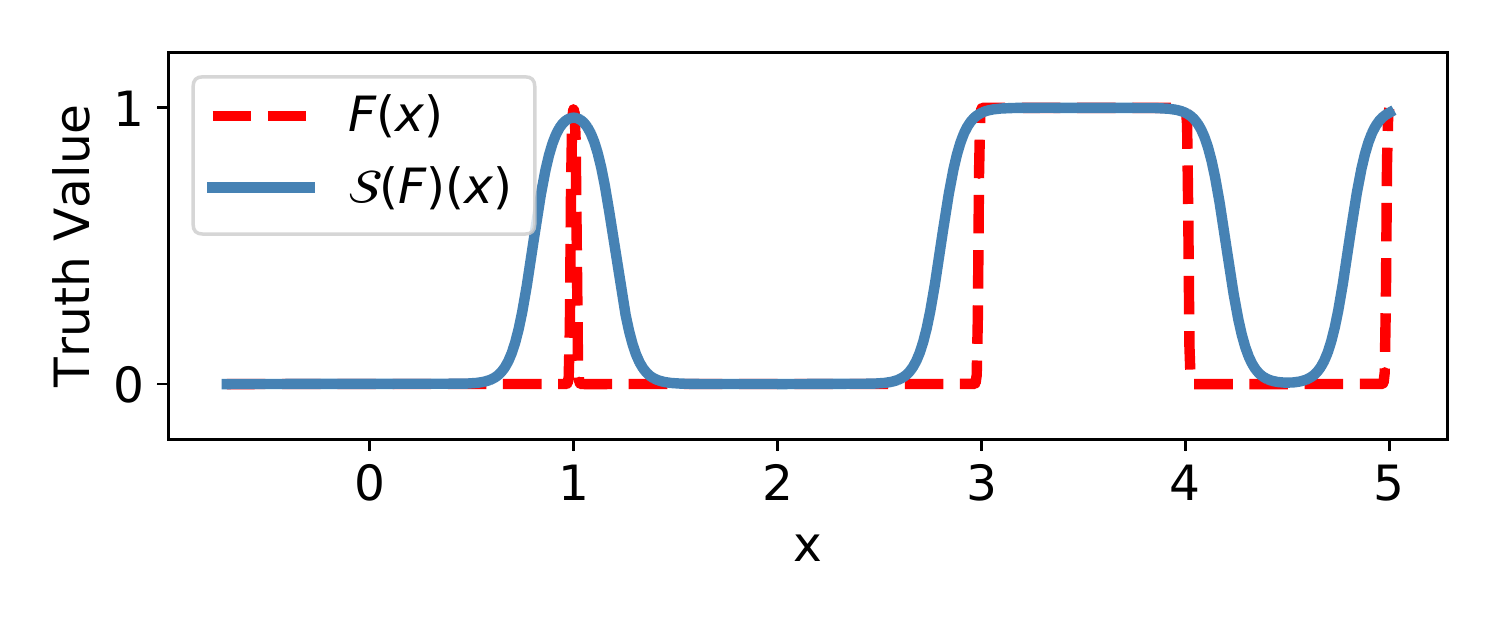}
 \setlength{\abovecaptionskip}{-5pt}
  \caption{\label{fig:ex_cln_plot}Plot of the formula $F(x) \triangleq (x = 1)  \lor (x \geq 5) \lor (x \geq 2 \land x \leq 3)$ and its associated CLN $M(x)$.}
\end{figure}


\section{Workflow}
\label{sec:workflow}
\begin{figure} 
\centering
\includegraphics[width=.95\linewidth, trim={0 0.0 0.0cm 0}]{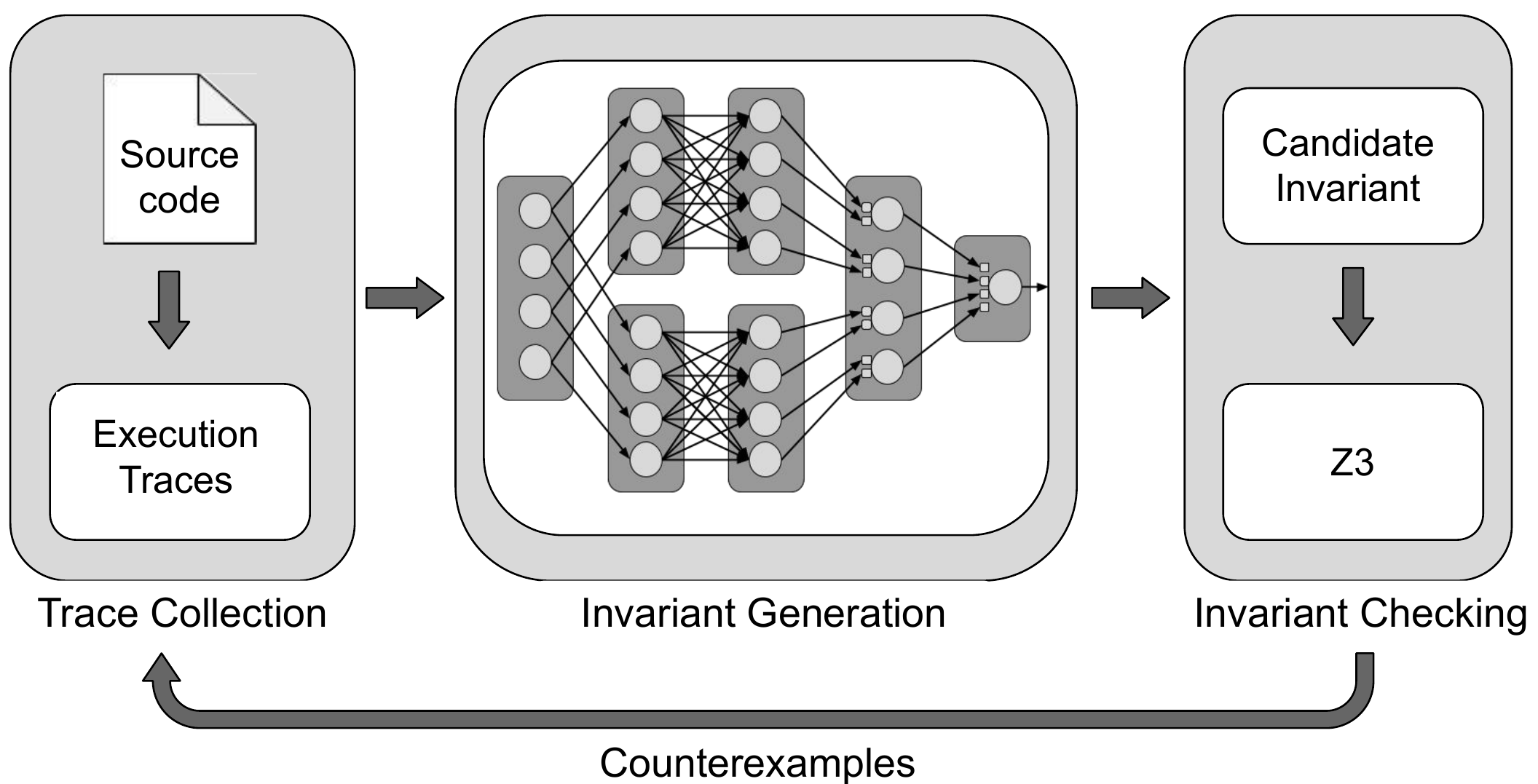}
\caption{Overview of method consisting of 3 phases: trace generation from source code file, G-CLN training, and invariant extraction followed by checking with Z3.
}
\setlength{\belowcaptionskip}{-20pt}
\label{fig:arch}
\end{figure}


\begin{figure}
\quad
\begin{subfigure}[b]{.5\columnwidth}
\begin{lstlisting}[style=myC,numbers=none,basicstyle=\linespread{0.9}\ttfamily\footnotesize]
// pre: (n >= 0)
a=0; s=1; t=1;
while (s<=n){
  log(a,s,t,n);
  a += 1;
  t += 2;
  s += t;
}
log(a,s,t,n);
\end{lstlisting}
\vspace{-5pt}
\caption{\label{fig:recordVars} Program instru-\\mented to log samples.}
\end{subfigure}
\hspace{-35pt}
\begin{subfigure}[b]{.5\columnwidth}
\vspace{10pt}
\begin{tabular}{cccccccc}
    \toprule
    1 & a & t & ... & $as$ & $t^2$ & $st$\\
    \midrule
    1 & 0 & 1 &    & 0 & 1 & 1 \\
    1 & 1 & 3 &    & 4 & 9 & 12 \\
    1 & 2 & 5 &    & 18 & 25 & 45 \\
    1 & 3 & 7 &    & 48 & 49 & 112 \\
    \bottomrule
\end{tabular}
\vspace{10pt}
\caption{\label{fig:workflowSamples} Sample data points generated with maximum degree of 2.}
\end{subfigure}
\caption{\label{fig:workflow} Training data generation for the program shown in Figure \ref{fig:ineq_ex}.}
\end{figure}
Figure \ref{fig:arch} illustrates our overall workflow for loop invariant inference. Our approach has three stages: (i) We first instrument the program and execute to generate trace data. (ii) We then construct and train a G-CLN model to fit the trace data. (iii) We extract a candidate loop invariant from the model and check it against a specification.

Given a program loop, we modify it to record variables for each iteration and then execute the program to generate samples. Figure \ref{fig:recordVars} illustrates this process for the \texttt{sqrt} program from Figure \ref{fig:ineq_ex}. The program has the input $n$ with precondition $(n \geq 0)$, so we execute with values $n=0,1,2,...$ for inputs in a set range. Then we expand the samples to all candidate terms for the loop invariant. By default, we enumerate all the monomials over program variables up to a given degree $maxDeg$, as shown in Figure \ref{fig:workflowSamples}. Our system can be configured to consider other non-linear terms like $x^y$. 

We then construct and train a G-CLN model using the collected trace data. We use the model architecture described in \S \ref{sec:gcln-arch}, with PBQUs for bounds learning using the procedure in \S \ref{sec:ineq_learning}.  After training the model, the SMT formula for the invariant is extracted by recursively descending through the model and extracting clauses whose gating parameters are above 0.5, as outlined in Algorithm \ref{alg:formulaExtraction}. On the \texttt{sqrt} program, the model will learn the invariant $\left(a ^2 \leq n\right) \land \left(t = 2a + 1\right) \land \left(su = ( a + 1 ) ^2\right)$.

Finally, if z3 returns a counterexample, we will incorporate it into the training data, and rerun the three stages with more generated samples. Our system repeats until a valid invariant is learned or times out.

\section{Theory}
\label{sec:theory}

In this section, we first present our gating construction for CLNs and prove gated CLNs are \emph{sound}
with regard to their underlying discrete logic. We then describe \emph{Piecewise Biased Quadratic Units}, a specific activation function construction for learning tight bounds on inequalities, and provide theoretical guarantees. Finally we present a technique to relax loop semantics and generate more samples when needed.

\subsection{Gated t-norms and t-conorms}
\label{sec:gating}

In the original CLNs \cite{cln2inv}, a formula template is required to learn the invariant. For example, to learn the invariant $(x+y=0) \lor (x-y=0)$, we have to provide the template $(w_1x+w_2y+b_1=0) \lor (w_3x+w_4y+b_2=0)$, which can be constructed as a CLN model to learn the coefficients. So, we have to know in advance whether the correct invariant is an atomic clause, a conjunction, a disjunction, or a more complex logical formula. 
To tackle this problem, we introduce gated t-norms and gated t-conorms.

Given a classic t-norm $T(x,y)=x \otimes y$, we define its associated gated t-norm as 
\[T_G(x,y;g_1,g_2) = (1 + g_1(x-1)) \otimes (1 + g_2(y-1))\]
Here $g_1,g_2 \in [0,1]$ are gate parameters indicating if $x$ and $y$ are activated, respectively.  The following equation shows the intuition behind gated t-norms.
\[T_G(x,y;g_1,g_2) = \left\{
\begin{array}{lcr}
    x \otimes y && g_1=1, g_2=1 \\
    x && g_1=1, g_2=0 \\
    y && g_1=0, g_2=1 \\
    1 && g_1=0, g_2=0
\end{array}\right.\]

Informally speaking when $g_1=1$, the input $x$ is activated and behaves as in the classic t-norm. When $g_1=0$, $x$ is deactivated and discarded. When $0<g_1<1$, the value of $g_1$ indicates how much information we should take from $x$. This pattern also applies for $g_2$ and $y$.

We can prove that $\forall g_1,g_2 \in [0,1]$, the gated t-norm is continuous and monotonically increasing with regard to $x$ and $y$, thus being well suited for training.

Like the original t-norm, the gated t-norm can be easily extended to more than two operands. In the case of three operands, we have the following: 
\begin{align*}
    T_G(x,y,z;g_1,g_2,g_3) = &\ (1 + g_1(x-1)) \otimes (1 + g_2(y-1)) \\
     & \otimes (1 + g_3(z-1))
\end{align*}

Using De Morgan's laws $x \oplus y = 1 - (1-x) \otimes (1-y)$, we define gated t-conorm as
\[T'_G(x,y;g_1,g_2) = 1 - (1 - g_1 x) \otimes (1 - g_2 y)\]
Similar to gated t-norms, gated t-conorms have the following property.
\[T'_G(x,y;g_1,g_2) = \left\{
\begin{array}{lcr}
    x \oplus y && g_1=1, g_2=1 \\
    x && g_1=1, g_2=0 \\
    y && g_1=0, g_2=1 \\
    0 && g_1=0, g_2=0
\end{array}\right.\]

Now we replace the original t-norms and t-conorms in CLN with our gated alternatives, which we diagram in Figure \ref{fig:gated_norm}. Figure \ref{fig:gating} demonstrates a gated CLN for representing an SMT formula. With the gated architecture, the gating parameters ${g_1, g_2}$ for each gated t-norm or gated t-conorm are made learnable during model training, such that the model can decide which input should be adopted and which should be discarded from the training data. This improves model flexibility and does not require a specified templates. 

\begin{figure}[bp]
\begin{subfigure}{0.65\linewidth}
    \includegraphics[width=\linewidth]{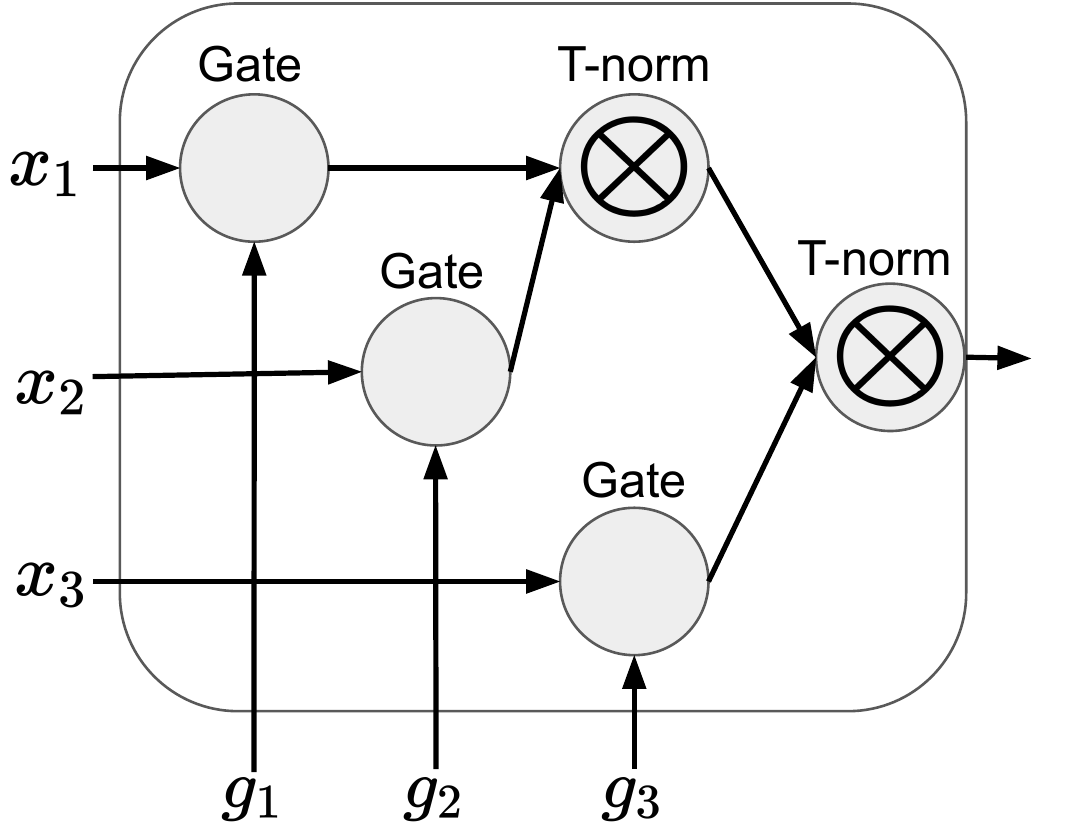}
\end{subfigure}
\setlength{\abovecaptionskip}{10pt}
\caption{Example of gated t-norm with three operands constructed from binary t-norms. The gated t-conorm is done similarly.}
\label{fig:gated_norm}
\setlength{\belowcaptionskip}{10pt}
\end{figure}

\begin{figure}[htbp]
\centering
\includegraphics[width=0.6\linewidth]{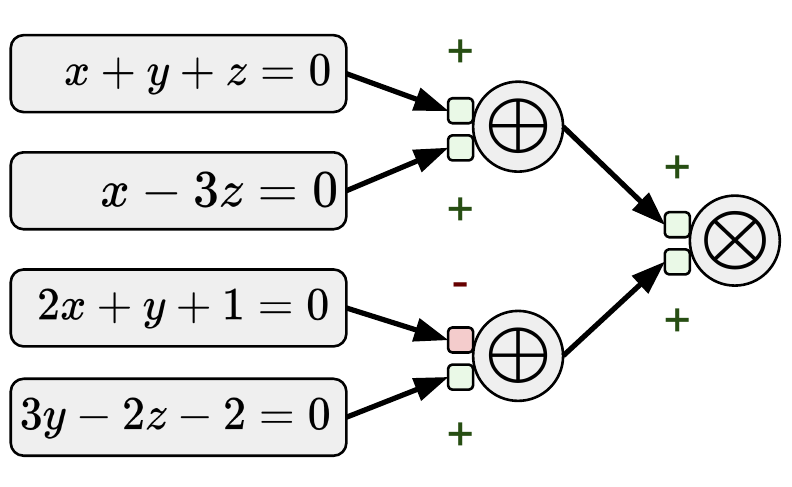}
\setlength{\abovecaptionskip}{10pt}
\caption{An instance of gated CLN. ``+'' means activated (g=1) and ``-'' means deactivated (g=0). The SMT formula learned is $ \left(3y-3z-2=0\right) \land \left((x-3z=0) \lor (x+y+z=0)\right)$.}
\label{fig:gating}
\setlength{\belowcaptionskip}{10pt}
\end{figure}

Now, we formally state the procedure to retrieve the SMT formula from a gated CLN model recursively in Algorithm~\ref{alg:formulaExtraction}. Abusing notation for brevity, $\mathcal{M}_i$ in line~1 represent the output node of model $\mathcal{M}_i$ rather than the model itself, and the same applies for line 8 and line 15. 
$BuildAtomicFormula$ in line 18 is a subroutine to extract the formula for a model with no logical connectives (e.g., retrieving $x+y+z=0$ in Figure \ref{fig:gating}). The linear weights which have been learned serve as the coefficients for the terms in the equality or inequality depending on the associated activation function. \revise{Finally, we need to round the learned coefficients to integers. We first scale the coefficients so that the maximum is $1$ and then round to the nearest rational number using a maximum possible denominator. We check if each rounded invariant fits all the training data and discard the invalid ones.}

\begin{algorithm}[t]
\caption{\textbf{Formula Extraction Algorithm.}}
\label{alg:formulaExtraction}
\begin{flushleft}
\textbf{Input:} A gated CLN model $\mathcal{M}$, with input nodes $\mathcal{X}=\{x_1,x_2,...,x_n\}$ and output node $p$. \\
\textbf{Output:} An SMT formula $F$ \\
\textbf{Procedure} ExtractFormula($\mathcal{M}$)
\end{flushleft}

\begin{algorithmic}[1]
\IF {$p = T_G(\mathcal{M}_1,...,\mathcal{M}_n;g_1,...,g_n)$} 
    \STATE{F := True} 
    \FOR {$i:=1$ \TO $n$ }{
        \IF {$g_i > 0.5$}
            \STATE {$F := F \land ExtractFormula(\mathcal{M}_i)$}
       \ENDIF
    }
    \ENDFOR
\ELSIF {$p = T_G'(\mathcal{M}_1,...,\mathcal{M}_n;g_1,...,g_n)$} 
    \STATE{F := False}
    \FOR {$i:=1$ \TO $n$} {
        \IF {$g_i > 0.5$}
            \STATE {$F := F \lor ExtractFormula(\mathcal{M}_i)$}
        \ENDIF
    }
    \ENDFOR
\ELSIF {$p = 1 - \mathcal{M}_1$}
    \STATE {$F := \neg ExtractFormula(\mathcal{M}_1)$}
\ELSE 
    \STATE{$F := BuildAtomicFormula(\mathcal{M})$}
\ENDIF
\end{algorithmic}
\end{algorithm}

In Theorem \ref{thm:soundness}, we will show that the extracted SMT formula is equivalent to the gated CLN model under some constraints. 
We first introduce a property of t-norms that is defined in the original CLNs~\cite{cln2inv}.\\

\parheader{Property 1.} $\forall t\ u,\ (t>0) \land (u>0) \implies (t \otimes u > 0)$. 

The product t-norm $x \otimes y = x \cdot y$, which is used in our implementation, has this property.

Note that the hyperparameters $c_1,c_2,\epsilon,\sigma$ in Theorem \ref{thm:soundness} will be formally introduced in \S{\ref{sec:ineqMapping}} and are unimportant here. One can simply see 
$$\lim_{\substack{c_1 \rightarrow 0, c_2 \rightarrow \infty \\ \sigma \rightarrow 0, \epsilon \rightarrow 0}} \mathcal{M}(x;c_1,c_2,\sigma,\epsilon)$$ 
as the model output $\mathcal{M}(x)$.

\begin{theorem}
\label{thm:soundness}
For a gated CLN model $\mathcal{M}$ with input nodes $\{x_1,x_2,...,x_n\}$ and output node $p$, if all gating parameters $\{g_i\}$ are either 0 or 1, then using the formula extraction algorithm, the recovered SMT formula $F$ is equivalent to the gated CLN model $\mathcal{M}$. That is, $\forall x \in R^n$, 
\begin{align}
    F(x) = True \iff \lim_{\substack{c_1 \rightarrow 0, c_2 \rightarrow \infty \\ \sigma \rightarrow 0, \epsilon \rightarrow 0}} \mathcal{M}(x;c_1,c_2,\sigma,\epsilon) = 1 \label{eq:soundness1} \\
    F(x) = False \iff \lim_{\substack{c_1 \rightarrow 0, c_2 \rightarrow \infty \\ \sigma \rightarrow 0, \epsilon \rightarrow 0}} \mathcal{M}(x;c_1,c_2,\sigma,\epsilon) = 0 \label{eq:soundness2}
\end{align}
as long as the t-norm in $\mathcal{M}$ satisfies Property 1.
\end{theorem}

\begin{proof}
We prove this by induction over the formula structure considering four cases: atomic, negation, T-norm, and T-conorm. For brevity, we sketch the T-norm case here and provide the full proof in Appendix A.

\vspace{5pt}
\parheader{T-norm Case.}  If $p = T_G(\mathcal{M}_1,...,\mathcal{M}_n;g_1,...,g_n)$,
which means the final operation in $\mathcal{M}$ is a gated t-norm, we know that for each submodel $\mathcal{M}_1,...,\mathcal{M}_n$ the gating parameters are all either 0 or 1. By the induction hypothesis, for each $\mathcal{M}_i$, using Algorithm \ref{alg:formulaExtraction}, we can extract an equivalent SMT formula $F_i$ satisfying Eq. (\ref{eq:soundness1})(\ref{eq:soundness2}). Then we can prove the full model $\mathcal{M}$ and the extracted formula $F$ also satisfy Eq. (\ref{eq:soundness1})(\ref{eq:soundness2}), using the induction hypothesis and the properties and t-norms.
\end{proof}

The requirement of all gating parameters being either 0 or 1 indicates that no gate is partially activated (e.g., $g_1=0.6$). Gating parameters between 0 and 1 are acceptable during model fitting but should be eliminated when the model converges. In practice this is achieved by gate regularization which will be discussed in \S{\ref{sec:gcln-arch}}.

Theorem \ref{thm:soundness} guarantees the soundness of the gating methodology with regard to discrete logic. Since the CLN architecture is composed of operations that are sound with regard to discrete logic, this property is preserved when gated t-norms and t-conorms are used in the network. 

Now the learnable parameters of our gated CLN include both linear weights $W$ as in typical neural networks, and the gating parameters $\{g_i\}$, so the model can represent a large family of SMT formulas. Given a training set $X$, when the gated CLN model $\mathcal{M}$ is trained to $\mathcal{M}(x)=1$ for all $x \in X$, then from Theorem \ref{thm:soundness} the recovered formula $F$ is guaranteed to hold true for all the training samples. That is, $\forall x \in X,F(x)=True$.

\subsection{Parametric Relaxation for Inequalities}
\label{sec:ineqMapping}

For learned inequality constraints to be useful in verification, they usually need to constrain the loop behavior as tightly as possible. In this section, we define a CLN activation function, the \emph{bounds learning activation}, which naturally learns tight bounds during training while maintaining the soundness guarantees of the CLN mapping to SMT. 
\begin{align}
\label{eq:inequalityMapping}
\C(t \geq u) \triangleq \left\{\begin{array}{lr}
    \frac{c_1^2}{(t-u)^2+c_1^2} & t < u \\
    \frac{c_2^2}{(t-u)^2+c_2^2} & t \geq u 
\end{array}\right.
\end{align}%
\noindent{}Here $c_1$ and $c_2$ are two constants. The following limit property holds.
\begin{align*}
\lim_{\substack{c_1 \rightarrow 0 \\ c_2 \rightarrow \infty}} S(t \geq u) = \left\{\begin{array}{lr}
0 & t < u \\
1 & t \geq u
\end{array}\right.
\end{align*}
Intuitively speaking, when $c_1$ approaches 0 and $c_2$ approaches infinity, $\C(x \geq 0)$ will approach 
the original semantic of predicate $\geq$. 
Figure \ref{fig:piesewiseIneq} provides an illustration of our parametric relaxation for $\geq$.

Compared with the sigmoid construction in the original CLNs (Figure \ref{fig:sigmoidIneq}), 
our parametric relaxation penalizes very large $x$, where $x \geq 0$ is absolutely correct but not very informative because the bound is too weak. In general, our piecewise mapping punishes data points farther away from the boundary, thus encouraging to learn a tight bound of the samples. On the contrary, the sigmoid construction encourages samples to be far from the boundary, resulting in loose bounds which are not useful for verification.

\begin{figure}[H]
\begin{subfigure}[b]{.46\columnwidth}
    \includegraphics[width=1.05\columnwidth]{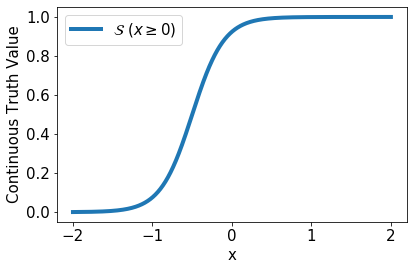}
    \caption{Plot of $\C(x \geq 0)$ with the CLNs' sigmoid construction.}
    \label{fig:sigmoidIneq}
\end{subfigure}
\quad
\begin{subfigure}[b]{.46\columnwidth}
    \includegraphics[width=1.05\columnwidth]{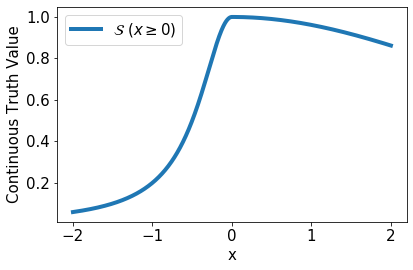}
    \caption{Plot of $\C(x \geq 0)$ with our piecewise construction.}
    \label{fig:piesewiseIneq}
\end{subfigure}
\setlength{\abovecaptionskip}{2pt}
\setlength{\belowcaptionskip}{-10pt}
\caption{\label{fig:ineqPlot}Comparison of the mapping $\C$  on $\geq$. The hyperparameters are $B=5$, $\epsilon=0.5$, $c_1=0.5$, and $c_2=5.$}
\end{figure}

\revise{Since the samples represent only a subset of reachable states from the program, encouraging a tighter bound may potentially lead to overfitting. However, we ensure soundness by later checking learned invariants via a solver. If an initial bound is too tight, we can incorporate counterexamples to the training data. Our empirical results show this approach works well in practice.}

Given a set of $n$ $k$-dimensional samples $\{[x_{11},...,x_{1k}],...,$ $[x_{n1}, ..., x_{nk}]\}$, where $x_{ij}$ denotes the value of variable $x_j$ in the $i$-th sample, we want to learn an inequality $w_1 x_1 + ... + w_k x_k + b \geq 0$ for these samples. The desirable properties of such an inequality is that it should be valid for all points, and have as tight as a fit as possible. Formally, we define a ``desired'' inequality as:
\begin{equation}
\label{eq:bounds_properties}
\begin{aligned}
\forall i \in \{1,...,n\}, w_1 x_{i1} + ... + w_k x_{ik} + b \geq 0\\
\exists j \in \{1,...,n\}, w_1 x_{j1} + ... + w_k x_{jk} + b = 0
\end{aligned}
\end{equation}

Our parametric relaxation for $\geq$ shown in Eq.~\ref{eq:inequalityMapping} can always learn an inequality which is very close to a ``desired'' one
with proper $c_1$ and $c_2$. Theorem~\ref{theorem:inequality} put this formally.

\begin{theorem}
\label{theorem:inequality}
Given a set of $n$ $k$-dimensional samples with the maximum L2-norm $l$, if $c_1 \leq 2l$ and $c_1 \cdot c_2 \geq 8 \sqrt{n} l^2$, and the weights are constrained as $\sum_{i=1}^k w_i^2 = 1$, then when the model converges, the learned inequality has distance at most $c_1/\sqrt{3}$ from a ``desired'' inequality.
\begin{proof} See Appendix B.
\end{proof}
\end{theorem}

Recall that $c_1$ is a small constant, so $c_1/\sqrt{3}$ can be considered as the error bound of inequality learning. Although we only proved the theoretical guarantee when learning a single inequality, our parametric relaxation for inequalities can be connected with other inequalities and equalities with conjunctions and disjunctions under a single CLN model. 

Based on our parametric relaxation for $\geq$, other inequality predicates can be defined accordingly.
\begin{align*}
    \C(t \leq u) \triangleq \left\{\begin{array}{lr}
    \frac{c_2^2}{(t-u)^2+c_2^2} & t < u \\
    \frac{c_1^2}{(t-u)^2+c_1^2} & t \geq u 
    \end{array}\right.
\end{align*}
\begin{align*}
    \C(t > u) \triangleq \C(t \geq (u + \epsilon)) &&
    \C(t < u) \triangleq \C(t \leq (u - \epsilon))
\end{align*}
where $\epsilon$ is a set small constant.

For our parametric relaxation, some axioms in classic logic just approximately rather than strictly hold (e.g., $t \leq u = \neg (t > u$)). They will strictly hold when $c_1 \rightarrow 0$ and $c_2 \rightarrow \infty$.

We reuse the Gaussian function as the parametric relaxation for equalities \cite{cln2inv}. Given a small constant $\sigma$,
\begin{align*}
    \C(t=u) \triangleq \exp{(-\frac{(t-u)^2}{2\sigma^2})}
\end{align*}

\subsection{Fractional Sampling}
In some cases, the samples generated from the original program are insufficient to learn the correct invariant due to dominating growth of some terms (higher-order terms in particular) or limited number of loop iterations. To generate more fine-grained yet valid samples, we perform \emph{Fractional Sampling} to relax the program semantics to continuous functions without violating the loop invariants by varying the initial value of program variables. The intuition is as follows.

Any numerical loop invariant $I$ can be viewed as a predicate over program variables \revise{$V$}  initialized with \revise{$V_0$} such that
\begin{align}
\label{eq:numericalInv}
    \forall V, \ V_0 \mapsto^* V \implies I(V)
\end{align}
where $V_0 \mapsto^* V$ means starting from initial values $V_0$ and executing the loop for 0 or more iterations ends with values $V$ for the variables.

Now we relax the initial values $X_0$ and see them as input variables $V_I$, which may carry arbitrary values. The new loop program will have variables $V \cup V_I$. Suppose we can learn an invariant predicate $I'$ for this new program, i.e.,
\begin{align}
    \label{eq:relaxedNumericalInv}
    \forall V_I \ V, \ V_I \mapsto^* V \implies I'(V_I \cup V)
\end{align}
Then let $V_I=V_0$, Eq. (\ref{eq:relaxedNumericalInv}) will become
\begin{align}
\label{eq:f'tof}
\forall V, \ V_0 \mapsto^* V \implies I'(V_0 \cup V)
\end{align}
Now $V_0$ is a constant, and $I'(V_0 \cup V)$ satisfies Eq. (\ref{eq:numericalInv}) thus being a valid invariant for the original program. In fact, if we learn predicate $I'$ successfully then we have a more general loop invariant that can apply for any given initial values.

Figure \ref{fig:fracSampling} shows how Fractional Sampling can generate more fine-grained samples with different initial values, making model fitting much easier in our data-driven learning system. The correct loop invariant for the program in Figure \ref{fig:ps4} is $$ (4 x = y^4 + 2 y^3 + y^2) \land (y \leq k) $$ To learn the equality part $(4 x = y^4 + 2 y^3 + y^2)$, if we choose $maxDeg=4$ and apply normal sampling, then six terms $\{1,y,y^2,y^3,y^4,x\}$ will remain after the heuristic filters in \S{\ref{sec:termSelection}}. 
Figure \ref{fig:withoutFracSampling}  shows a subset of training samples without Fractional Sampling (the column of term 1 is omitted).

\begin{figure}
\quad
\begin{subfigure}{.35\columnwidth}
\begin{lstlisting}[style=myC,numbers=none,basicstyle=\linespread{0.9}\ttfamily\footnotesize]
//pre: x = y = 0 
//     /\ k >= 0
while (y < k) {
  y++;
  x += y * y * y;
}
//post: 4x == k^2
//    * (k + 1)^2 
\end{lstlisting}
\vspace{7pt}
\caption{\label{fig:ps4} The ps4 program \\ in the benchmark.}
\end{subfigure}
\hspace{10pt}
\begin{subfigure}{.5\columnwidth}
\begin{tabular}{cccccc}
    \toprule
    x & y & $y^2$ & $y^3$ & $y^4$ \\
    \midrule
    0 & 0 & 0 & 0 & 0 \\
    1 & 1 & 1 & 1 & 1 \\
    9 & 2 & 4 & 8 & 16 \\
    36 & 3 & 9 & 27 & 81 \\
    100 & 4 & 16 & 64 & 256 \\
    225 & 5 & 25 & 125 & 625 \\
    \bottomrule
\end{tabular}
\vspace{2pt}
\caption{\label{fig:withoutFracSampling} Training data generated \\ without Fractional Sampling.}
\end{subfigure}
\begin{subfigure}{\columnwidth}
\vspace{12pt}
\resizebox{\columnwidth}{!}{
\begin{tabular}{rrrrrrrrrr}
    \toprule
    x & y & $y^2$ & $y^3$ & $y^4$ & $x_0$ & $y_0$ & $y_0^2$ & $y_0^3$ & $y_0^4$  \\
    \midrule
    -1 & -0.6 & 0.36 & -0.22 & 0.13 & -1 & -0.6 & 0.36 & -0.22 & 0.13 \\
    -0.9 & 0.4 & 0.16 & 0.06 & 0.03 & -1 & -0.6 & 0.36 & -0.22 & 0.13 \\
    1.8& 1.4 & 1.96 & 2.74 & 3.84 & -1 & -0.6 & 0.36 & -0.22 & 0.13 \\
    0 & -1.2 & 1.44 & -1.73 & 2.07 & 0 & -1.2 & 1.44 & -1.73 & 2.07 \\
    0 & -0.2 & 0.04 & -0.01 & 0.00 & 0 & -1.2 & 1.44 & -1.73 & 2.07 \\
    0.5 & 0.8 & 0.64 & 0.52 & 0.41 & 0 & -1.2 & 1.44 & -1.73 & 2.07 \\
    \bottomrule
\end{tabular}
}
\caption{\revise{Training data generated with fractional sampling.}}
\label{fig:withFracSampling}
\end{subfigure}
\caption{\label{fig:fracSampling} An example of Fractional Sampling.}
\end{figure}


When $y$ becomes larger, the low order terms $1,y,$ and $y^2$ become increasingly negligible because they are significantly smaller than the dominant terms $y^4$ and $x$. In practice we observe that the coefficients for $x^4$ and $y$ can be learned accurately but not for $1,y,y^2$. To tackle this issue, we hope to generate more samples around $y=1$ where all terms are on the same level. Such samples can be easily generated by feeding more initial values around $y=1$ using Fractional Sampling. Table \ref{fig:withFracSampling} shows some generated samples from $x_0=-1, y_0=-0.6$ and $x_0=0, y_0=-1.2$.

Now we have more samples where terms are on the same level, making the model easier to converge to the accurate solution. Our gated CLN model can correctly learn the relaxed invariant $4 x - y^4 - 2 y^3 - y^2 - 4 x_0 + y_0^4 + 2 y_0^3 + y_0^2 = 0$. Finally we return to the exact initial values $x_0=0, y_0=0$, and the correct invariant for the original program will appear $4 x - y^4 - 2 y^3 - y^2 = 0$.

Note that for convenience, in Eq. (\ref{eq:numericalInv})(\ref{eq:relaxedNumericalInv})(\ref{eq:f'tof}), we assume all variables are initialized in the original program and all are relaxed in the new program. However, the framework can easily extends to programs with uninitialized variables, or we just want to relax a subset of initialized variables. \revise{Details on how fractional sampling is incorporated in our system are provided in \S\ref{sec:frac-samp-impl}.}

\section{Nonlinear Invariant Learning}
\label{sec:methodology}

In this section, we describe our overall approach for nonlinear loop invariant inference. We first describe our methods for stable CLN training on nonlinear data. We then give an overview of our model architecture and how we incorporate our inequality activation function to learn inequalities. Finally, we show how we extend our approach to also learn invariants that contain external functions. 

\subsection{Stable CLN Training on Nonlinear Data}

Nonlinear data causes instability in CLN training due to the large number of terms and widely varying magnitudes it introduces. We address this by modifying the CLN architecture to normalize both inputs and weights on a forward execution. We then describe how we implement term dropout, which helps the model learn precise SMT coefficients.

\subsubsection{Data Normalization}

Exceedingly large inputs cause instability and prevent the CLN model from converging to precise SMT formulas that fit the data. We therefore modify the CLN architecture such that it rescales its inputs so the L2-norm equals a set value $l$. In our implementation, we used $l=10$. 
\begin{table}[b]
\caption{Training data after normalization for the program in Figure \ref{fig:ineq_ex}, which computes the integer square root.}
\begin{tabular}{cccccccc}
    \toprule
    1 & a & t & ... & $as$ & $t^2$ & $st$\\
    \midrule
    0.70 & 0 & 0.70 & & 0 & 0.70 & 0.70 \\
    0.27 & 0.27 & 0.81 & & 1.08 & 2.42 & 3.23 \\
    0.13 & 0.25 & 0.63 & & 2.29 & 3.17 & 5.71 \\
    0.06 & 0.19 & 0.45 & & 3.10 & 3.16 & 7.23 \\
    \bottomrule
\end{tabular}
\label{tbl:dataNorm}
\end{table}

We take the program in Figure \ref{fig:ineq_ex} as an example. The raw samples before normalization is shown in Figure \ref{fig:workflowSamples}. The monomial terms span in a too wide range, posing difficulty for network training. With data normalization, each sample (i.e., each row) is proportionally rescaled to L2-norm 10. 
The normalized samples are shown in Table \ref{tbl:dataNorm}.

Now the samples occupy a more regular range. Note that data normalization does not violate model correctness. If the original sample $(t_1,t_2,...t_k)$ satisfies the equality $w_1t_1 + w_2t_2 + ... + w_kt_k = 0$ (note that $t_i$ can be a higher-order term), so does the normalized sample and vice versa. The same argument applies to inequalities.

\subsubsection{Weight Regularization}

For both equality invariant $w_1 x_1 + ... + w_m x_m + b = 0$ and inequality invariant $w_1 x_1 + ... + w_m x_m + b \geq 0$, $w_1=...=w_m=b=0$ is a true solution. To avoid learning this trivial solution, we require at least one of $\{w_1, ..., w_m\}$ is non-zero. A more elegant way is to constrain the $L^p$-norm
of the weight vector to constant 1. In practice we choose L2-norm as we did in Theorem \ref{theorem:inequality}. The weights are constrained to satisfy
\[w_1^2 + ... + w_m^2 = 1\]

\subsubsection{Term Dropout}
\label{sec:termSelection}

Given a program with three variables $\{x,y,z\}$ and $maxDeg=2$, we will have ten candidate terms $\{1,x,y,z,x^2,y^2,z^2,xy,xz,$ $yz\}$.
The large number of terms poses difficulty for invariant learning, and the loop invariant in a real-world program is unlikely to contain all these terms. We use two methods to select terms. First the growth-rate-based heuristic in \cite{sharma2013data} is adopted to filter out unnecessary terms. Second we apply a random dropout to discard terms before training.

Dropout is a common practice in neural networks to avoid overfitting and improve performance. Our dropout is randomly predetermined before the training, which is different from the typical weight dropout in deep learning~\cite{srivastava2014dropout}. Suppose after the growth-rate-based filter, seven terms $\{1,x,$ $y,z,x^2,y^2,xy\}$ remain. 
Before the training, each input term to a neuron may be
discarded with probability $p$.

The importance of dropout is twofold. First it further reduces the number of terms in each neuron. Second it encourages G-CLN to learn more simple invariants. For example, if the desired invariant is $(x-y-1=0) \land (x^2-z=0)$, then a neuron may learn their linear combination (e.g., $2x-2y-2+x^2-z=0$) which is correct but not human-friendly. If the term $x$ is discarded in one neuron then that neuron may learn $x^2-z=0$ rather than $2x-2y-2+x^2-z=0$. 
Similarly, if the terms $x^2$ and $xy$ are discarded in another neuron, then that neuron may learn $x-y-1=0$. Together, the entire network consisting of both neurons will learn the precise invariant.

Since the term dropout is random, a neuron may end up having no valid invariant to learn (e.g., both $x$ and $x^2$ are discarded in the example above). But when gating (\S\ref{sec:gating}) is adopted, this neuron will be deactivated and remaining neurons may still be able to learn the desired invariant. More details on gated model training will be provided in \S{\ref{sec:gcln-arch}}.



\subsection{Gated CLN Invariant Learning}

Here we describe the Gated CLN architecture and how we incorporate the 
bounds learning activation function to learn general nonlinear loop invariants. 

\subsubsection{Gated CLN Architecture}
\label{sec:gcln-arch}

\begin{figure}[tbp]
\centering
\includegraphics[width=\linewidth]{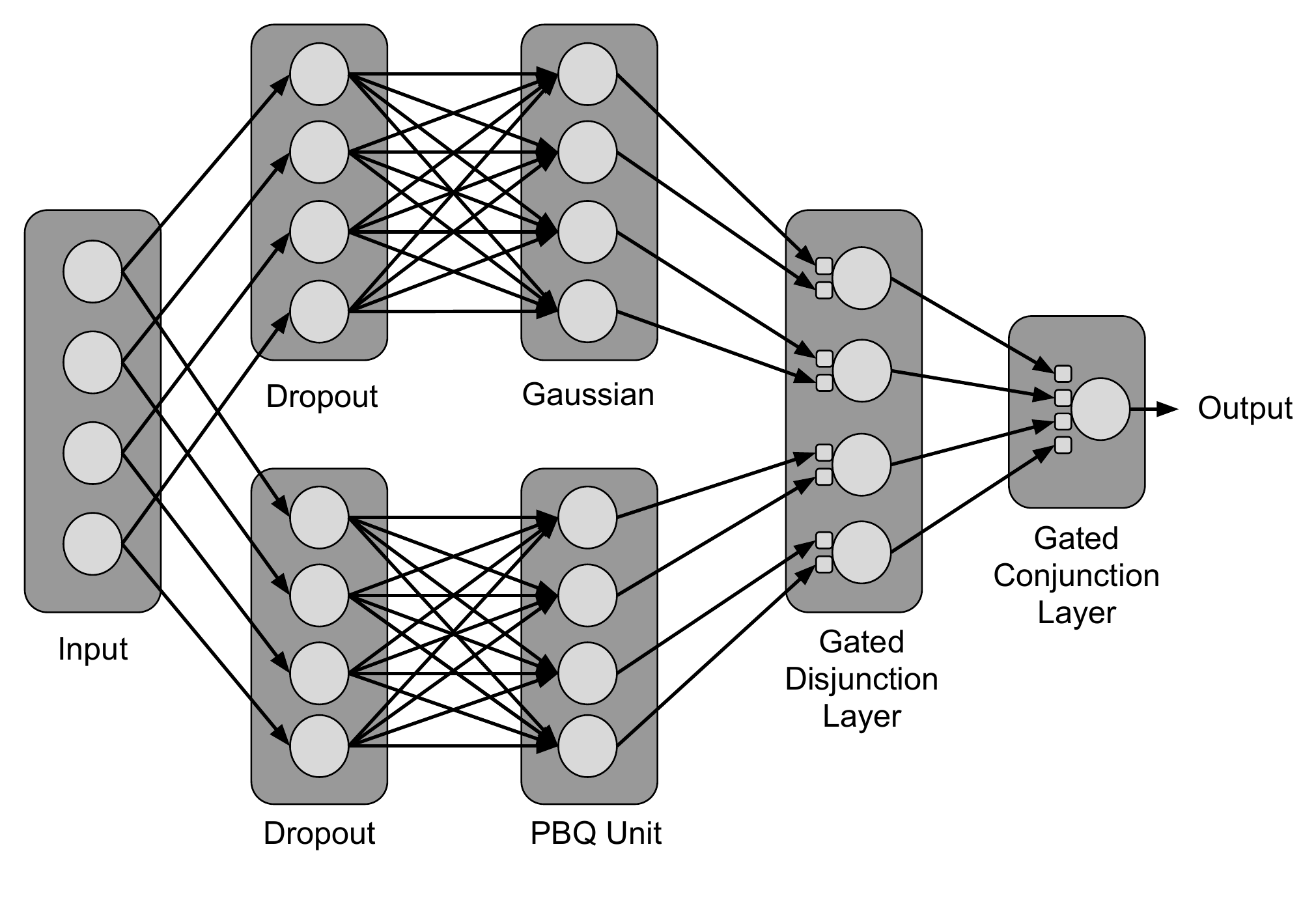}
\setlength{\abovecaptionskip}{-5pt}
\caption{Diagram of G-CLN model. Additional disjunction and conjunction layers may be added to learn more complex SMT formulas.}
\setlength{\belowcaptionskip}{-10pt}
\label{fig:nn_arch}
\end{figure}
\parheader{Architecture. }
In \S{3.1}, we introduce gated t-norm and gated t-conorm, and illustrate how they can be integrated in CLN architecture. Theoretically, the gates can cascade to many layers, while in practice, we use a gated t-conorm layer representing logical OR followed by a gated t-norm layer representing logical AND, as shown in Figure \ref{fig:nn_arch}. The SMT formula extracted from such a gated CLN architecture will be in conjunctive normal form (CNF). In other words, G-CLN is parameterized by m and n, where the underlying formula can be a conjunction of up to m clauses, where each clause is a disjunction of n atomic clauses. In the experiments we set m=10, n=2.


\parheader{Gated CLN Training. }
For the sake of discussion, consider a gated t-norm with two inputs. We note that the gating parameters $g_1$ and $g_2$ have an intrinsic tendency to become 0 in our construction. When $g_1=g_2=0$, $T_G(x,y;g_1,g_2)=1$ regardless of the truth value of the inputs $x$ and $y$. So when training the gated CLN model, we apply regularization on $g_1,g_2$ to penalize small values.
Similarly, for a gated t-conorm, the gating parameters $g_1,g_2$ have an intrinsic tendency to become 1 because $x \oplus y$ has a greater value than $x$ and $y$. To resolve this we apply regularization pressure on $g_1,g_2$ to penalize close-to-1 values.

In the general case, given training set $X$ and \revise{gate regularization parameters $\lambda_1,\lambda_2$}, the model will learn to minimize the following loss function with regard to the linear weights $W$ and gating parameters $G$,
\begin{align*}
\mathcal{L}(X;W,G) = &\sum_{x \in X} (1-\mathcal{M}(\mathbf{x};W,G))\\
&+ \lambda_1\sum_{g_i \in T_G} (1-g_i) + \lambda_2\sum_{g_i \in T_G'} g_i
\end{align*}

\noindent By training the G-CLN with this loss formulation, the model tends to 
learn a formula $F$ satisfying each training sample (recall $F(x)=True \Leftrightarrow \mathcal{M}(x)=1$ in \S{\ref{sec:gating}}).
Together, gating and regularization prunes off poorly learned clauses, while preventing the network from pruning off too aggressively.
When the training converges, all the gating parameters will be very close to either 1 or 0, indicating the participation of the clause in the formula. The invariant is recovered using Algorithm 1.



\subsubsection{Inequality Learning}
\label{sec:ineq_learning}

\begin{figure}
\centering
\begin{subfigure}{\linewidth}
\includegraphics[width=\textwidth]{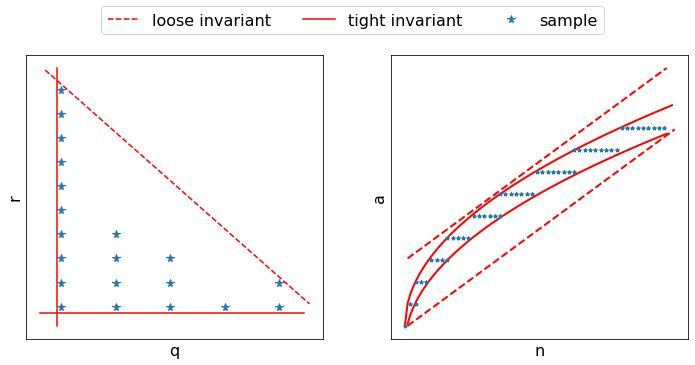}
\end{subfigure}
\;
\begin{subfigure}{0.4\linewidth}
\caption{\label{fig:bounds_final} Learned inequality bounds.}
\end{subfigure}
\;\;
\begin{subfigure}{0.4\linewidth}
\caption{\label{fig:bounds_sqrt} Learned inequality bounds on \texttt{sqrt}.}
\end{subfigure}
\setlength{\abovecaptionskip}{10pt}
\caption{\label{fig:bounds_examples}Examples of 2 dimensional bound fitting. }
\setlength{\belowcaptionskip}{-60pt}
\end{figure}

Inequality learning largely follows the same procedure as equality learning with two differences. First, we use the PBQU activation (i.e., the parametric relaxation for $\geq$) introduced in \S{\ref{sec:ineqMapping}}, instead of the Gaussian activation function (i.e., the parametric relaxation for $=$). This difference is shown in Figure \ref{fig:nn_arch}. As discussed in \S{\ref{sec:ineqMapping}}, the PBQU activation will learn tight inequalities rather than loose ones.



\revise{Second, we structure the dropout on inequality constraints to consider all possible combinations of variables up to a set number of terms and maximum degree (up to 3 terms and 2nd degree in our evaluation). We then train the model following the same optimization used in equality learning, and remove constraints that do not fit the data based on their PBQU activations after the model has finished training.}

When extracting a formula from the model we remove poorly fit learned bounds that have PBQU activations below a set threshold.
As discussed in \S\ref{sec:ineqMapping}, PBQU activations penalizes points that are farther from its bound. The tight fitting bounds in Figures \ref{fig:bounds_final} and \ref{fig:bounds_sqrt} with solid red lines have PBQU activations close to 1, while loose fitting bounds with dashed lines have PBQU activations close to 0. \revise{After selecting the best fitting bounds, we check against the loop specification and remove any remaining constraints that are unsound. If the resulting invariant is insufficient to prove the postcondition, the model is retrained using the counterexamples generated during the specification check.}

\subsection{External Function Calls}
In realistic applications, loops are not entirely whitebox and may contain calls to external functions for which the signature is provided but not the code body. In these cases, external functions may also appear in the loop invariant. To address these cases, when an external function is present, we sample it during loop execution. 
To sample the function, we execute it with all combinations of variables in scope during sampling that match its call signature.

For example, the function $gcd: \mathbb{Z} \times \mathbb{Z} \rightarrow \mathbb{Z}$, for greatest common divisor, is required in the invariant for four of the evaluation problems that compute either greatest common divisor or least common multiple: (\texttt{egcd2}, \texttt{egcd3}, \texttt{lcm1}, and \texttt{lcm2}).
In practice, we constrain our system to binary functions, 
but it is not difficult to utilize static-analysis to extend the support to more complex external function calls. This procedure of constructing terms containing external function calls
is orthogonal to our training framework. 

\subsection{Fractional Sampling Implementation}
\label{sec:frac-samp-impl}

\revise{We apply fractional sampling on a per-program basis when we observe the model is unable to learn the correct polynomial from the initial samples. We first sample on $0.5$ intervals, then $0.25$, etc. until the model learns a correct invariant. We do not apply fractional sampling to variables involved in predicates and external function calls, such as gcd. In principle, predicate constraints can be relaxed to facilitate more general sampling. We will investigate this in future work.}

\revise{Among all the programs in our evaluation, only two of them, ps5 and ps6, require fractional sampling. For both of them sampling on $0.5$ intervals is sufficient to learn the correct invariant, although more fine grained sampling helps the model learn a correct invariant more quickly. The cost associated with fractional sampling is small ($<5s$).}

\section{Evaluation}
\label{sec:evaluation}

We evaluate our approach on NLA, a benchmark of common numerical algorithms with nonlinear invariants. We first perform a comparison with two prior works, Numinv and PIE, that use polynomial equation solving and template learning respectively. We then perform an ablation of the methods we introduce in this paper. Finally, we evaluate the stability of our approach against a baseline CLN model.

\vspace{2mm}

\parheader{Evaluation Environment.} The experiments described in this section were conducted on an Ubuntu 18.04 server with an Intel XeonE5-2623 v4 2.60GHz CPU, 256Gb of memory, and an Nvidia GTX 1080Ti GPU.

\vspace{2mm}

\parheader{System Configuration. }
We implement our method with the PyTorch Framework and use the Z3 SMT solver to validate the correctness of the inferred loop invariants. \revise{For the four programs involving greatest common divisors, we manually check the validity of the learned invariant since gcd is not supported by z3.} We use a G-CLN model with the CNF architecture described in \S\ref{sec:gcln-arch}, with a conjunction of ten clauses, each with up to two literals. \revise{We use adaptive regularization on the CLN gates. $\lambda_1$ is set to $(1.0,0.999,0.1)$, which means that $\lambda_1$ is initialized as $1.0$, and is multiplied by $0.999$ after each epoch until it reaches the threshold $0.1$. Similarly, $\lambda_2$ is set to $(0.001,1.001,0.1)$. We try three maximum denominators, $(10,15,30)$, for coefficient extraction in \S\ref{sec:gating}. For the parametric relaxation in \S\ref{sec:ineqMapping}, we set $\sigma=0.1, c_1 = 1, c_2 = 50$. The default dropout rate in \S\ref{sec:termSelection} is $0.3$, and will decrease by $0.1$ after each failed attempt until it reaches $0$. We use the Adam optimizer with learning rate $0.01$, decay $0.9996$, and max epoch $5000$.}

\subsection{Polynomial Invariant Dataset}
We evaluate our method on a dataset of programs that require nonlinear polynomial invariants \cite{nguyen2012NLA}. The problems in this dataset represent various numerical algorithms ranging from modular division and greatest common denominator ($gcd$) to computing geometric and power series. These algorithms involve up to triply nested loops and sequential loops, which we handle by predicting all the requisite invariants using the model before checking their validity. We sample within input space of the whole program just as we do with single loop problems and annotate the loop that a recorded state is associated with. The invariants involve up to $6^{th}$ order polynomial and up to thirteen variables. 

 \begin{table}
 \centering
 \small
 \caption{ \label{tbl:poly_eval}Table of problems requiring nonlinear polynomial invariant from NLA dataset. 
 We additionally tested Code2Inv on the same problems as PIE and it fails to solve any within 1 hour. Numinv results are based on Table 1 in \cite{nguyen2017numinv}. G-CLN solves 26 of 27 problems with an average execution time of 53.3 seconds.}
 \begin{tabular}{p{1.4cm}p{1.0cm}p{0.9cm}p{0.8cm}p{1.1cm}p{1.0cm}}
 \toprule
 \textbf{Problem} & 
 \textbf{Degree} & 
 \textbf{\# Vars} &
 \textbf{PIE} &
 {\specialcellbold{NumInv}} &  
 {\specialcellbold{G-CLN}} \\
 \midrule
divbin & 2 & 5 & - & \cmark & \cmark \\
cohendiv & 2 & 6 & - & \cmark & \cmark \\
mannadiv  & 2 & 5 & \xmark & \cmark & \cmark \\
hard  & 2 & 6 & - & \cmark & \cmark \\
sqrt1  & 2 & 4 & - & \cmark & \cmark \\
dijkstra & 2 & 5 & - & \cmark & \cmark \\
cohencu & 3 & 5 & - & \cmark & \cmark \\
egcd & 2 & 8 & - & \cmark & \cmark \\
egcd2 & 2 & 11 & - & \xmark & \cmark \\
egcd3 & 2 & 13 & - & \xmark & \cmark \\
prodbin & 2 & 5 & - & \cmark & \cmark \\
prod4br & 3 & 6 & \xmark & \cmark & \cmark \\
fermat1 & 2 & 5 & - & \cmark & \cmark \\
fermat2 & 2 & 5 & - & \cmark & \cmark \\
freire1 & 2 & 3 & - & \xmark & \cmark \\
freire2 & 3 & 4 & - & \xmark & \cmark \\
knuth & 3 & 8 & - & \cmark & \xmark \\
lcm1 & 2 & 6 & - & \cmark & \cmark \\
lcm2 & 2 & 6 & \xmark & \cmark & \cmark \\
geo1 & 2 & 5 & \xmark & \cmark & \cmark \\
geo2 & 2 & 5 & \xmark & \cmark & \cmark \\
geo3 & 3 & 6 & \xmark & \cmark & \cmark \\
ps2 & 2 & 4 & \xmark & \cmark & \cmark \\
ps3 & 3 & 4 & \xmark & \cmark & \cmark \\
ps4 & 4 & 4 & \xmark & \cmark & \cmark \\
ps5 & 5 & 4 & - & \cmark & \cmark \\
ps6 & 6 & 4 & - & \cmark & \cmark \\
 \bottomrule
 \end{tabular}
 \end{table}

 \begin{table}
 \centering
 \small
 \caption{ \label{tbl:abl}Table with ablation of various components in the G-CLN model. Each column reports which problems can be solved with G-CLN when that feature ablated.
 }
 \resizebox{0.95\columnwidth}{!}{
 \begin{tabular}{p{1.1cm}p{0.8cm}p{0.8cm}p{0.8cm}p{1.2cm}p{0.9cm}}
 \toprule
 \textbf{Problem} & 
 \textbf{Data} & 
 \textbf{Weight} &
 {\specialcellbold{Drop-}} 
 &  {\specialcellbold{Frac.}}
 &  {\specialcellbold{Full}}
 \\
 &   \textbf{Norm.} & \textbf{Reg.} &\textbf{out} 
 &  {\specialcellbold{Sampling}}
 & \textbf{Method}
\\
 \midrule
divbin    & \xmark & \xmark  & \cmark & \cmark & \cmark\\
cohendiv  & \xmark & \xmark  & \xmark & \cmark & \cmark\\
mannadiv  & \xmark & \xmark  & \cmark & \cmark & \cmark\\
hard      & \xmark & \xmark  & \cmark & \cmark & \cmark\\
sqrt1     & \xmark & \xmark  & \xmark & \cmark & \cmark\\
dijkstra  & \xmark & \xmark  & \cmark & \cmark & \cmark\\
cohencu   & \xmark & \xmark  & \cmark & \cmark & \cmark\\
egcd      & \xmark & \cmark  & \xmark & \cmark & \cmark\\
egcd2     & \xmark & \cmark  & \xmark & \cmark & \cmark\\
egcd3     & \xmark & \cmark  & \xmark & \cmark & \cmark\\
prodbin   & \xmark & \cmark  & \cmark & \cmark & \cmark\\
prod4br   & \xmark & \cmark  & \cmark & \cmark & \cmark\\
fermat1   & \xmark & \cmark  & \cmark & \cmark & \cmark\\
fermat2   & \xmark & \cmark  & \cmark & \cmark & \cmark\\
freire1   & \xmark & \cmark  & \cmark & \cmark & \cmark\\
freire2   & \xmark & \cmark  & \xmark & \cmark & \cmark\\
knuth     & \xmark & \xmark  & \xmark & \xmark & \xmark\\
lcm1      & \xmark & \cmark  & \cmark & \cmark & \cmark\\
lcm2      & \xmark & \cmark  & \cmark & \cmark & \cmark\\
geo1      & \xmark & \cmark  & \cmark & \cmark & \cmark\\
geo2      & \xmark & \cmark  & \cmark & \cmark & \cmark\\
geo3      & \xmark & \cmark  & \cmark & \cmark & \cmark\\
ps2       & \cmark & \xmark  & \cmark & \cmark & \cmark\\
ps3       & \xmark & \xmark  & \cmark & \cmark & \cmark\\
ps4       & \xmark & \xmark  & \cmark & \cmark & \cmark\\
ps5       & \xmark & \xmark  & \cmark & \xmark & \cmark\\
ps6       & \xmark & \xmark  & \cmark & \xmark & \cmark\\
 \bottomrule
 \end{tabular}
 }
 \end{table}
 
 \vspace{2mm}
 
\parheader{Performance Comparison.}
Our method is able to solve 26 of the 27 problems as shown in Table \ref{tbl:poly_eval}, while NumInv solves 23 of 27. Our average execution time was 53.3 seconds, which is a minor improvement to NumInv who reported 69.9 seconds. We also evaluate LoopInvGen (PIE) on a subset of the simpler problems which are available in a compatible format\footnote{LoopInvGen uses the SyGuS competition format, which is an extended version of smtlib2.}. It was not able to solve any of these problems before hitting a 1 hour timeout. 
In Table \ref{tbl:poly_eval}, we indicate solved problems with \cmark, unsolved problems with \xmark, \space and problems that were not attempted with $-$.

The single problem we do not solve, \texttt{knuth}, is the most difficult problem from a learning perspective.
The invariant for the problem, $(d^2*q - 4*r*d + 4*k*d - 2*q*d + 8*r == 8*n) \land (mod (d,  2) == 1)$, is one of the most complex in the benchmark. 
Without considering the external function call to $mod$ (modular division), there are already 165 potential terms of degree at most 3, nearly twice as many as next most complex problem in the benchmark, making it difficult to learn a precise invariant with gradient based optimization. We plan to explore better initialization and training strategies to scale to complex loops like \texttt{knuth} in future work.

Numinv is able to find the equality constraint in this invariant because its approach is specialized for equality constraint solving. However, we note that NumInv only infers octahedral inequality constraints and does not in fact infer the nonlinear and 3 variable inequalities in the benchmark.

We handle the $mod$ binary function successfully in \texttt{fermat1} and \texttt{fermat2} indicting the success of our model in supporting external function calls. Additionally, for four problems (\texttt{egcd2}, \texttt{egcd3}, \texttt{lcm1}, and \texttt{lcm2}), we incorporate the $gcd$ external function call as well. 
 
 \subsection{Ablation Study}

We conduct an ablation study to demonstrate the benefits gained by the normalization/regularization techniques and term dropouts as well as fractional sampling. 
Table \ref{tbl:abl} notes that data normalization is crucial for nearly all the problems, especially for preventing high order terms from dominating the training process. 
Without weight regularization, the problems which involve inequalities over multiple variables cannot be solved; 7 of the 27 problems cannot be solved without dropouts, which help avoid the degenerate case where the network learns repetitions of the same atomic clause. Fractional sampling helps to solve high degree ($5^{th}$ and $6^{th}$ order) polynomials as the distance between points grow fast.

\subsection{Stability}

We compare the stability of the gated CLNs with standard CLNs as proposed in \cite{cln2inv}. Table \ref{tbl:stability} shows the result. We ran the two CLN methods without automatic restart 20 times per problem and compared the probability of arriving at a solution. We tested on the example problems described in \cite{cln2inv} with disjunction and conjunction of equalities, two problems from Code2Inv, as well as ps2 and ps3 from NLA. As we expected, our regularization and gated t-norms vastly improves the stability of the model as clauses with poorly initialized weights can be ignored by the network. We saw improvements across all the six problems, with the baseline CLN model having an average convergence rate of $58.3\%$, and the G-CLN converging $97.5\%$ of the time on average. 
\subsection{Linear Invariant Dataset}

We evaluate our system on the Code2Inv benchmark \cite{si2018learning} of 133 linear loop invariant inference problems with source code and SMT loop invariant checks. We hold out 9 problems shown to be theoretically unsolvable in \cite{cln2inv}. Our system finds correct invariants for all remaining 124 theoretically solvable problems in the benchmark in under 30s.

\begin{figure}[t]
 \begin{table}[H]
 \centering
 \small
 \caption{ \label{tbl:stability}Table comparing the stability of CLN2INV with our method. The statistics reported are over 20 runs per problem with randomized initialization. }
 \begin{tabular}{c*{2}{c}}
 \toprule
 \textbf{Problem} & 
 \textbf{Convergence Rate} & 
 \textbf{Convergence Rate} 
 \\
 &   \textbf{of CLN} & \textbf{of G-CLN} 
\\
 \midrule
Conj Eq        & 75\% & 95\%  \\
Disj Eq        & 50\% & 100\%  \\
Code2Inv 1     & 55\% & 90\%  \\
Code2Inv 11    & 70\% & 100\%  \\
ps2            & 70\% & 100\%  \\
ps3            & 30\% & 100\%  \\
 \bottomrule
 \end{tabular}
 \end{table}
 \vspace{-30pt}
 \end{figure}

\section{Related Work}
\label{sec:related}

\parheader{Numerical Relaxations.}
Inductive logic programming (ILP) has been used to learn a logical formula consistent with a set of given data points.
More recently, efforts have focused on differentiable relaxations of ILP  for learning \citep{kimmig2012short, yang2017differentiable, evans2018learning,payani2019inductive} or program synthesis \cite{si2019synthesizing}. Other recent efforts have used formulas as input to graph and recurrent nerual networks to solve Circuit SAT problems and identify Unsat Cores \cite{amizadeh2018learning, selsam2018learning, selsam2019guiding}. FastSMT also uses a neural network select optimal SMT solver strategies \cite{balunovic2018learning}. In contrast, our work relaxes the semantics of the SMT formulas allowing us to learn SMT formulas. 

\vspace{2mm}



\parheader{Counterexample-Driven Invariant Inference.} There is a long line of work to learn loop invariants based on counterexamples. 
ICE-DT uses decision tree learning and leverages counterexamples which violate the inductive verification condition \cite{garg2016learning, zhu2018chc, garg2014ice}. Combinations of linear classifiers have been applied to learning CHC clauses \cite{zhu2018chc}.

A state-of-the-art method, LoopInvGen (PIE) learns the loop invariant using enumerative synthesis to repeatedly add data consistent clauses to strengthen the post-condition until it becomes inductive \cite{sharma2013verification, padhi2016data, loopinvgen2017}. 
For the strengthening procedure, LoopInvGen uses PAC learning, a form of boolean formula learning, to learn which combination of candidate atomic clauses is consistent with the observed data. In contrast, our system learn invariants from trace data. 

\vspace{2mm}

\parheader{Neural Networks for Invariant Inference.} Recently, 
neural networks have been applied to loop invariant inference. Code2Inv combines graph and recurrent neural networks to model the program graph and learn from counterexamples~\cite{si2018learning}.
In contrast, CLN2INV uses CLNs to learn SMT formulas for invariants directly from program data \cite{cln2inv}. We also use CLNs but incorporate gating and other improvements to be able to learn general nonlinear loop invariants.


\vspace{2mm}

\parheader{Polynomial Invariants.} There have been efforts to utilize abstract interpretation to discover polynomial invariants~\cite{rodriguez2007generating, rodriguez2004automatic}. 
More recently, \revise{Compositional Recurrance Analysis (CRA) performs analysis on abstract domain of transition formulas, but relies on over approximations that prevent it from learning sufficient invariants~\cite{farzan2015compositional, kincaid2017compositional, kincaid2017non}.} 
Data-driven methods based on linear algebra such as \textit {Guess-and-Check} are able to learn polynomial equality invariants accurately \cite{sharma2013data}. 
Guess-and-check learns equality invariants by using the polynomial kernel, 
but it cannot learn disjunctions and inequalities,
which our framework supports natively.

\textit{NumInv}~\cite{nguyen2012using,nguyen2014dig, nguyen2017numinv} uses the polynomial kernel but also learns octahedral inequalities.
NumInv sacrifices soundness for performance by replacing Z3 with KLEE, a symbolic executor, and in particular, treats invariants which lead to KLEE timeouts as valid. Our method instead is sound and learns more general inequalities than NumInv.

\vspace{-0.4cm}
\section{Conclusion}
\label{sec:conclusion}

We introduce G-CLNs, a new gated neural architecture that can learn general nonlinear loop invariants. We additionally introduce Fractional Sampling, a method that soundly relaxes program semantics to perform dense sampling, and PBQU activations, which naturally learn tight inequality bounds for verification.  We evaluate our approach on a set of 27 polynomial loop invariant inference problems and solve 26 of them, 3 more than prior work, as well as improving convergence rate to $97.5\%$ on quadratic problems, a $39.2\%$ improvement over CLN models.


\section*{Acknowledgements}
The authors are grateful to our shepherd, 
Aditya Kanade, and the anonymous reviewers for  valuable feedbacks that improved this paper significantly.
This work is sponsored in part by NSF grants CNS-18-42456,
CNS-18-01426, CNS-16-17670, CCF-1918400; ONR grant N00014-17-1-2010; an ARL Young
Investigator (YIP) award; an NSF CAREER award; a Google Faculty
Fellowship; a Capital One Research Grant;  a J.P. Morgan Faculty
Award;
a Columbia-IBM Center Seed Grant Award;
and a Qtum Foundation Research Gift. Any opinions, findings, conclusions, or recommendations that are
expressed herein are those of the authors, and do not necessarily
reflect those of the US Government, ONR, ARL, NSF, Google, Capital One
J.P. Morgan, IBM, or Qtum.

\bibliography{main}

\appendix
\section{Proof for Theorem 3.1}

\parheader{Property 1.} $\forall t\ u, (t>0) \land (u>0) \implies (t \otimes u > 0)$. 

\parheader{Theorem 3.1. } For a gated CLN model $\mathcal{M}$ with input nodes $\{x_1,x_2,...,x_n\}$ and output node $p$, if all gating parameters $\{g_i\}$ are either 0 or 1, then using the formula extraction algorithm, the recovered SMT formula $F$ is equivalent to the gated CLN model $\mathcal{M}$. That is, $\forall x \in R^n$, 
\begin{align}
    F(x) = True \iff \lim_{\substack{c_1 \rightarrow 0, c_2 \rightarrow \infty \\ \sigma \rightarrow 0, \epsilon \rightarrow 0}} \mathcal{M}(x;c_1,c_2,\sigma,\epsilon) = 1 \label{eq:appsoundness1} \\
    F(x) = False \iff \lim_{\substack{c_1 \rightarrow 0, c_2 \rightarrow \infty \\ \sigma \rightarrow 0, \epsilon \rightarrow 0}} \mathcal{M}(x;c_1,c_2,\sigma,\epsilon) = 0 \label{eq:appsoundness2}
\end{align}
as long as the t-norm in $\mathcal{M}$ satisfies Property 1.

\begin{proof}

We prove this by induction.

\parheader{T-norm Case.} If $p = T_G(\mathcal{M}_1,...,\mathcal{M}_n;g_1,...,g_n)$
which means the final operation in $\mathcal{M}$ is a gated t-norm, we know that for each submodel $\mathcal{M}_1,...,\mathcal{M}_n$ the gating parameters are all either 0 or 1. According to the induction hypothesis, for each $\mathcal{M}_i$, using Algorithm 1, we can extract an equivalent SMT formula $F_i$ satisfying Eq. (\ref{eq:appsoundness1})(\ref{eq:appsoundness2}). Now we prove the full model $\mathcal{M}$ and the extracted formula $F$ satisfy Eq. (\ref{eq:appsoundness1}). The proof for Eq. (\ref{eq:appsoundness2}) is similar and omitted for brevity. For simplicity we use $\mathcal{M}(x)$ to denote $\lim_{\substack{\sigma \rightarrow 0, \epsilon \rightarrow 0 \\ c_1 \rightarrow 0, c_2 \rightarrow \infty}} \mathcal{M}(x;c_1,c_2,\sigma,\epsilon)$. 

Now the proof goal becomes $F(x)=True \iff \mathcal{M}(x)=1$, and the induction hypothesis is
\begin{equation}
\label{eq:soundnessIH}
\forall i \in \{1,...,n\},\ F_i(x)=True \iff \mathcal{M}_i(x)=1
\end{equation}

From line 2-5 in Algorithm 1, we know that the extracted formula $F$ is the conjunction of a subset of formulas $F_i$ if the $F_i$ is activated ($g_i > 0.5$). Because in our setting all gating parameters are either 0 or 1, then a logical equivalent form of $F$ can be derived.
\begin{equation}
    \label{eq:soundnessFAnd}
    F(x) = \land_{i=1}^n ((g_i = 0) \lor F_i(x))
\end{equation}

Recall we are considering the t-norm case where
\begin{align*}
 & \mathcal{M}(x) = T_G(\mathcal{M}_1(x),...,\mathcal{M}_n(x);g_1,...,g_n) = \\
 & (1 + g_1(\mathcal{M}_1(x)-1)) \otimes ... \otimes (1 + g_n(\mathcal{M}_n(x)-1))
\end{align*}
Using the properties of a t-norm in \S{2.2} we can prove that 
\begin{align*}
 \mathcal{M}(x) = 1 \iff \forall i \in \{1,...,n\},\ (1 + g_i(\mathcal{M}_i(x)-1)) = 1
\end{align*}
Because the gating parameter $g_i$ is either 0 or 1, we further have
\begin{equation}
\label{eq:soundnessMxMix}
 \mathcal{M}(x) = 1 \iff \forall i \in \{1,...,n\},\ g_i = 0 \lor \mathcal{M}_i = 1
\end{equation}

Combining Eq. (\ref{eq:soundnessIH})(\ref{eq:soundnessFAnd})(\ref{eq:soundnessMxMix}), we will finally have
\begin{align*}
     & F(x)=True \iff \forall i \in \{1,...,n\},\ (g_i = 0) \lor F_i(x) = True \\
     & \iff \forall i \in \{1,...,n\},\ (g_i = 0) \lor \mathcal{M}_i(x)=1 \\
     & \iff \mathcal{M}(x) = 1
\end{align*}

\parheader{T-conorm Case.} If $p = T_G'(\mathcal{M}_1,...,\mathcal{M}_n;g_1,...,g_n)$ which means the final operation in $\mathcal{M}$ is a gated t-conorm, similar to the t-norm case, for each submodel $\mathcal{M}_i$ we can extract its equivalent SMT formula $F_i$ using Algorithm 1 according to the induction hypothesis. Again we just prove the full model $\mathcal{M}$ and the extracted formula $F$ satisfy Eq. (\ref{eq:appsoundness1}).

From line 7-10 in Algorithm 1, we know that the extracted formula $F$ is the disjunction of a subset of formulas $F_i$ if the $F_i$ is activated ($g_i > 0.5$). Because in our setting all gating parameters are either 0 or 1, then a logical equivalent form of $F$ can be derived.
\begin{equation}
    \label{eq:soundnessFOr}
    F(x) = \lor_{i=1}^n ((g_i = 1) \land F_i(x))
\end{equation}

Under the t-conorm case, we have
\begin{align*}
 & \mathcal{M}(x) = T_G'(\mathcal{M}_1(x),...,\mathcal{M}_n(x);g_1,...,g_n) = \\
 & 1 - (1 - g_1\mathcal{M}_1(x)) \otimes ... \otimes (1 - g_n\mathcal{M}_n(x))
\end{align*}
Using Property 1 we can prove that 
\begin{align*}
 \mathcal{M}(x) = 1 \iff \exists i \in \{1,...,n\},\ 1 - g_i\mathcal{M}_i(x) = 0
\end{align*}
Because the gating parameter $g_i$ is either 0 or 1, we further have
\begin{equation}
\label{eq:soundnessMxMix2}
 \mathcal{M}(x) = 1 \iff \exists i \in \{1,...,n\}, g_i = 1 \land \mathcal{M}_i = 1
\end{equation}

Combining Eq. (\ref{eq:soundnessIH})(\ref{eq:soundnessFOr})(\ref{eq:soundnessMxMix2}), we will finally have
\begin{align*}
     & F(x)=True \iff \exists i \in \{1,...,n\},\ (g_i = 1) \land F_i(x) = True \\
     & \iff \exists i \in \{1,...,n\},\ (g_i = 1) \land \mathcal{M}_i(x)=1 \\
     & \iff \mathcal{M}(x) = 1
\end{align*}

\parheader{Negation Case.} If $p = 1 - \mathcal{M}_1$ which means the final operation is a negation, from the induction hypothesis we know that using Algorithm 1 we can extract an SMT formula $F_1$ from submodel $\mathcal{M}_1$ satisfying Eq. (\ref{eq:appsoundness1})(\ref{eq:appsoundness2}). From line 11-12 in Algorithm 1 we know that the extracted formula for $\mathcal{M}$ is $F=\neg F_1$. Now we show such an $F$ satisfy Eq. (\ref{eq:appsoundness1})(\ref{eq:appsoundness2}).
\begin{align*}
     & F(x) = True \iff F_1(x) = False \iff M_1(x) = 0 \\ 
     & \iff M(x) = 1 \\
     & F(x) = False \iff F_1(x) = True \iff M_1(x) = 1 \\ 
     & \iff M(x) = 0
\end{align*}

\parheader{Atomic Case.} In this case, the model $\mathcal{M}$ consists of only a linear layer and an activation function, with no logical connectiveness. This degenerates to the atomic case for the ungated CLN in \cite{cln2inv} where the proof can simply be reused.
\end{proof}

\section{Proof for Theorem 3.2}

Recall our continuous mapping for inequalities.
\begin{align}
\label{eq:appinequalityMapping}
\C(t \geq u) \triangleq \left\{\begin{array}{lr}
    \frac{c_1^2}{(t-u)^2+c_1^2} & t < u \\
    \frac{c_2^2}{(t-u)^2+c_2^2} & t \geq u 
\end{array}\right.
\end{align}

\parheader{Theorem 3.2. } Given a set of $n$ $k$-dimensional samples with the maximum L2-norm $l$, if $c_1 \leq 2l$ and $c_1 \cdot c_2 \geq 8 \sqrt{n} l^2$, and the weights are constrained as $w_1^2 + ... + w_k^2 = 1$, then when the model converges, the learned inequality has distance at most $c_1/\sqrt{3}$ from a 'desired' inequality.

\begin{proof}
The model maximize the overall continuous truth value, so the reward function is 
\begin{align*}
    R(w_1,...,w_k,b) = \sum_{i=1}^n f(w_1x_{i1} + ... + w_k x_{ik} + b)
\end{align*}
where $f(x)=\C(x \geq 0)$. We first want to prove that, if there exists a point that breaks the inequality with distance more than $c_1 / \sqrt{3}$,
then $\frac{\partial R}{\partial b} > 0$, indicating the model will not converge here. Without loss of generality, we assume the first point $\{x_{11},...x_{1k}\}$ breaks the inequality, i.e.,
\begin{align}
\label{eq:ineqOnePointBreak}
    w_1 x_{i1} + ... + w_k x_{ik} + b + C1/\sqrt{3} < 0
\end{align}
We will consider the following two cases.

(${\romannumeral 1}$) All the points breaks the inequality, i.e.
\begin{align*}
    \forall i \in \{1,...,n\}, w_1 x_{i1} + ... + w_k x_{ik} + b < 0
\end{align*}
From Eq. (\ref{eq:appinequalityMapping}), it is easy to see that
\begin{align}
\label{eq:ineqDerivativeSign}
    \left\{\begin{array}{cc}
        f'(x) > 0 & x < 0 \\
        f'(x) < 0 & x > 0
    \end{array} \right.
\end{align}
Then we have
\begin{align*}
    \frac{\partial R}{\partial b} = \sum_{i=1}^n f'(w_1x_{i1} + ... + w_k x_{ik} + b) > \sum_{i=1}^n 0 = 0
\end{align*}

($\romannumeral 2$) At least one point, say, $\{x_{j1}, ..., x_{jk}\}$, satisfies the inequality.
\begin{align}
\label{eq:ineqOnePointSatisfy}
    w_1 x_{j1} + ... + w_k x_{jk} + b \geq 0
\end{align}
From Cauchy–Schwarz inequality, $\forall i \in \{1,...,n\}$ we have
\begin{align*}
    & (w_1 x_{i1} + ... + w_k x_{ik})^2 \leq (w_1^2 + ... + w_k^2)(x_{i1}^2 + ... + x_{ik}^2) \\
    & = 1 \cdot (x_{i1}^2 + ... + x_{ik}^2) \leq l^2
\end{align*}
So
\begin{align}
\label{eq:ineqWxbound}
    \forall i \in \{1,...,n\},\ -l \leq w_1 x_{i1} + ... + w_k x_{ik} \leq l
\end{align}
Combining Eq. (\ref{eq:ineqOnePointBreak})(\ref{eq:ineqOnePointSatisfy})(\ref{eq:ineqWxbound}), a lower bound and an upper bound of $b$ can be obtained
\begin{align}
\label{eq:ineqBBound}
    -l \leq b < l - c_1/\sqrt{3}
\end{align}
Using basic calculus we can show that $f'(x)$ is strictly increasing in $(-\infty, -c_1/\sqrt{3})$ and $(c_2/\sqrt{3}, +\infty)$, and strictly decreasing in $(-c_1/\sqrt{3}, c_2/\sqrt{3})$. Combining Eq. (\ref{eq:ineqDerivativeSign})(\ref{eq:ineqWxbound})(\ref{eq:ineqBBound}), we have
\begin{align}
\label{eq:ineqGeneralLowerBound}
    \begin{split}
    & f'(w_1 x_{i1} + ... + w_k x_{ik} + b) \geq f'(|w_1 x_{i1} + ... + w_k x_{ik} + b|) \\
    & \geq f'(|w_1 x_{i1} + ... + w_k x_{ik}| + |b|) \geq f'(l+l) = f'(2l)
    \end{split}
\end{align}
The deduction of Eq. (\ref{eq:ineqGeneralLowerBound}) requires $2l \leq C2/\sqrt{3}$, which can be obtained from the two known conditions $c_1 \leq 2l$ and $c_1c_2 \geq 8 \sqrt{n} l^2$. 

Eq. (\ref{eq:ineqGeneralLowerBound}) provides a lower bound of the derivative for any point. For the point that breaks the inequality in Eq. (\ref{eq:ineqOnePointBreak}), since $f'(x)$ is strictly increasing in $(\infty, -c_1/\sqrt{3})$, we can obtain a stronger lower bound
\begin{align}
    f'(w_1x_{11} + ... + w_k x_{1k} + b) \geq f'(-l-l) = f'(-2l)
\end{align}
Put it altogether, we have
\begin{align*}
    & \frac{\partial R}{\partial b} =  f'(w_1x_{11} + ... + w_k x_{1k} + b) + \\
    & \sum_{i=2}^n f'(w_1x_{i1} + ... + w_k x_{ik} + b) \\
    & \geq f'(-2l) + (n-1) f'(2l) > f'(-2l) + n f'(2l) \\
    & > \frac{4l}{c_1^2 (1 + \frac{4l^2}{c_1^2})^2 (1 + \frac{4l^2}{c_2^2})^2} \cdot (1 - \frac{64nl^4}{c_1^2 c_2^2})
\end{align*}
Some intermediate steps are omitted. Because we know $c_1 c_2 \geq 8 \sqrt{n} l^2$, finally we have $\frac{\partial R}{\partial b} > 0$.

Now we have proved that no point can break the learned inequality more than $c_1/\sqrt{3}$. We need to prove at least one point is on or beyond the boundary. We prove this by contradiction. Suppose all points satisfy
\[w_1 x_{i1} + ... + w_k x_{ik} + b > 0\]
Then we have 
\[\frac{\partial R}{\partial b} = \sum_{i=1}^n f'(w_1x_{i1} + ... + w_k x_{ik} + b) < \sum_{i=1}^n 0 = 0\]
So the model does not converge, which concludes the proof.
\end{proof}

\end{document}